\documentclass{article}

\usepackage{fullpage}

\usepackage{graphicx} 
\usepackage[utf8]{inputenc} 
\usepackage[T1]{fontenc}    
\usepackage[hidelinks]{hyperref}       
\usepackage{url}            
\usepackage{booktabs}       
\usepackage{amsfonts}       
\usepackage{nicefrac}       
\usepackage{microtype}      
\usepackage{xcolor}         
\usepackage[numbers]{natbib}

\usepackage[noend,ruled,noline,linesnumbered]{algorithm2e}

\usepackage{amsmath,amssymb,amsthm,mathtools,mathdots,nicefrac,tikz,enumitem,url,bm}

\usepackage{thmtools,thm-restate}

\usepackage{soul}
\usepackage{setspace}
\usepackage{dsfont}
\usepackage{subcaption}
\usepackage{adjustbox}
\usepackage{cleveref}

\usepackage{xcolor}
\usepackage{color-edits}
\addauthor{eb}{blue}

\declaretheorem[style=definition,numberwithin=section]{theorem}
\declaretheorem[style=definition,numberwithin=section]{lemma}
\declaretheorem[style=definition,numberwithin=section]{proposition}

\declaretheorem[style=definition]{property}

\newtheorem*{theorem*}{Theorem}
\newtheorem*{property*}{Property}
\newtheorem{reduction}{Reduction}

\DeclareMathOperator*{\argmax}{argmax}
\DeclareMathOperator*{\argmin}{argmin}

\newcommand{\process}{\textsc{Process}}

\newcommand{\quickstream}{\textsc{QuickStream}}
\newcommand{\quickmax}{\textsc{QuickSwap}}
\newcommand{\algnm}{\textsc{QuickSwapNM}}
\newcommand{\quickpm}{\textsc{QuickSwapPM}}
\newcommand{\excha}{\textsc{ExchangeCandidate} }

\newcommand{\epsi}{ \varepsilon }
\newcommand{\oh}[1]{\mathcal{O}\left( #1 \right)}
\newcommand{\ohs}[2]{\mathcal{O}_{#1}\left( #2 \right)}
\newcommand{\opt}{\textsc{Opt}}
\newcommand{\mon}{\textsc{MSM}}
\newcommand{\nmon}{\textsc{GSM}}
\newcommand{\pmat}{\textsc{MSPM}}

\newcommand{\uni}{\mathcal N}
\newcommand{\ind}{\mathcal I}
\newcommand{\mat}{\mathcal M}

\SetArgSty{textnormal}

\title{Submodular Maximization in Exactly $n$ Queries}
\author{Eric Balkanski  \\
Columbia University \\
New York City, New York \\
\texttt{eb3224@columbia.edu}
\and Steven DiSilvio \\
MIT \\
Cambridge, Massachusetts \\
\texttt{disilvio@mit.edu }
\and 
  Alan Kuhnle \\                                     
  Texas A\&M University\\                                                 
  College Station, Texas \\                                                        
  \texttt{kuhnle@tamu.edu}
  \and
  ChunLi Peng\thanks{Authors are ordered alphabetically.} \\                                     
  Texas A\&M University\\                                                 
  College Station, Texas \\                                                        
  \texttt{chunli.peng@tamu.edu}
}

\date{August 21, 2024}
\begin{document}

\maketitle

\begin{abstract}
In this work, we study the classical problem of maximizing a submodular function subject to a matroid constraint. We develop deterministic algorithms that are very parsimonious with respect to querying the submodular function, for both the case when the submodular function is monotone and the general submodular case. In particular, we present a $1/4$ approximation algorithm for the monotone case that uses exactly one query per element, which gives the same total number of queries $n$ as the number of queries required to compute the maximum singleton. For the general case, we present a constant factor approximation algorithm that requires $2$ queries per element, which is the first algorithm for this problem with linear query complexity in the size of the ground set. 
\end{abstract}
\textbf{Version 2 update.} The same algorithm and $1/4$ approximation result for the monotone case were previously obtained by~\citet{dutting2023fully}. At the time of writing this manuscript we were not aware of this other paper. We do generalize the algorithm to the non-monotone case, for which we achieve a constant approximation, and to the family of $p$-matchoid constraints, for which we achieve a $1/(4p)$ approximation. Due to the base algorithm for the monotone case being identical as in~\citet{dutting2023fully}, we view the technical contribution of this manuscript limited.

\section{Introduction}

\paragraph{Submodular optimization.} Many objectives that we aim to optimize in machine learning, such as coverage, diversity, and entropy, satisfy the diminishing returns property required for a function to be submodular. Submodular maximization algorithms are thus employed in applications such as  document summarization~\citep{lin2011class}, influence maximization in networks~\citep{kempe2003maximizing}, recommender systems~\citep{badanidiyuru2014streaming}, and feature selection~\citep{das2011submodular}.
A fundamental problem in this field is to maximize a monotone submodular function under a matroid constraint, which we refer to as \mon{}. 
 For this problem, the celebrated greedy algorithm of~\citet{fisher1978analysis}  achieves a $1/2$ approximation. This $1/2$ approximation was improved by~\citet{calinescu2011maximizing} to $1-1/e$ with a continuous greedy algorithm, which is the best approximation achievable with polynomially many queries~\citep{nemhauser1978best}. The broader problem of maximizing a, not necessarily monotone, submodular function under a cardinality constraint, which we refer to as \nmon{}, is also known to  admit constant factor approximation algorithms \citep{lee2009non}. 

\paragraph{Fast algorithms for \mon{}.} Due to the many applications of submodular maximization   over massive datasets,
a major focus has been to develop fast algorithms for submodular maximization (see, e.g., \citep{Badanidiyuru2014b,Mirzasoleiman2015,Ene2019,Breuer2020,Kuhnle2021a,Li2022}).
 Since the time required to perform function evaluations usually dominates other parts of the
computation, the speed of an algorithm is often measured by its query complexity, i.e., its number of function evaluations. In particular, the query complexity of the greedy algorithm is $kn$, where $k$ is the rank of the matroid and $n$ is size of the ground set.
\citet{Badanidiyuru2014b} improved this query complexity  with an  algorithm with  $\oh{ (n/\epsi) \log( k / \epsi )}$ query complexity that achieves a $1/2 - \epsi$ approximation for \mon{}. Subsequently, \citet{Chakrabarti2015b} achieved a linear query complexity with a
$1/4$-approximation algorithm that makes at most $2n$ queries. 
On the hardness side, \citet{Kuhnle2021a} showed that $n$ queries are required to obtain an approximation  of $1/2 + \epsi$, even for the special case of finding the maximum singleton.

\paragraph{Fast algorithms for \nmon{}.}  If the objective function is not monotone,
there are no known algorithms with linear query complexity. \citet{lee2009non} obtained the first constant factor approximation algorithm for \nmon{}, with an algorithm that has query complexity  $\tilde{O}(n^4)$. An algorithm of \citet{Kuhnle2019e}
achieves a ratio $1/4 - \epsi$ in $\oh{ \frac{n}{\epsi} \log \frac{k}{\epsi} }$ queries for the special case of a size constraint, 
and this algorithm was extended to handle a matroid constraint
by \citet{Han2020}, keeping the same ratio and query complexity. Since the rank $k$ may be as large as $n$, the query complexity of these algorithms is superlinear in the size of the ground set. \newline

The main question we ask in this paper is whether the best-known query complexity for these two problems can be improved while achieving a constant factor approximation guarantee.

\begin{center}
\emph{ What is the query complexity of achieving a constant factor approximation for \mon{} and \nmon{}?}
\end{center}

\paragraph{Our results.} Our first main result shows that $n$ queries are sufficient to achieve a constant factor approximation with a deterministic algorithm for \mon{}.

\begin{theorem*}
There is a deterministic $1/4$-approximation algorithm for \mon{} with query complexity $n$.    
\end{theorem*}

This result improves the query complexity of
\citet{Chakrabarti2015b} from $2n$ to $n$, while maintaining the
same approximation. We emphasize that our algorithm, called \quickmax{}, does a pass over the elements and performs a \emph{single} query per element. The main idea needed to perform only a single query per element $e$ requires maintaining an  \emph{infeasible} solution and evaluating the marginal contribution of $e$ to this infeasible solution, whereas the vast majority of algorithms for submodular maximization only maintain a feasible solution. In fact, there have been lower bounds for submodular maximization that make the assumption that the algorithm only queries feasible sets~\citep{norouzi2018beyond,kupfer2020adaptive}, and this has been considered a natural assumption. Whether our result can be achieved by only querying feasible sets is an interesting question for future work.  This $n$ query complexity matches the number of queries required to find the element $e$ with maximum singleton value $f(\{e\})$~\cite{Kuhnle2021a}.
We generalize this algorithm to the $p$-matchoid constraint, and obtain:
\begin{theorem*} 
    There is a deterministic $1/(4p)$-approximation algorithm for \pmat{} with query complexity $n$.    
\end{theorem*}

Our third main result is the first constant factor approximation algorithm with linear query complexity for \nmon{}.
\begin{theorem*}
There is a deterministic  $1/(6 + 4\sqrt{2}) \approx 1/ 11.66$-approximation algorithm for \nmon{} with query complexity $2n$.    
\end{theorem*}

The previous best query complexity achieved by a constant factor approximation algorithm for \nmon{} is $\oh{\frac{n}{\epsi} \log \frac{k}{\epsi}}$. In addition to achieving the first linear query complexity algorithm for this problem, we emphasize that another  benefit of our algorithm is that the query complexity does not depend on large constants, which is important for the practicality of the algorithm. We emphasize that the algorithm and its analysis for \nmon{} build on the \quickmax{} algorithm for  \mon{}. A similar approach could obtain a constant factor approximation for the general $p$-matchoid constraint.

Finally, we empirically demonstrate the practicality of \quickmax{}. On real and synthetic datasets, it always achieves an improved number of queries and a similar objective value as the algorithm of \citet{Chakrabarti2015b} with $2n$ queries. Compared to the lazy greedy algorithm and the algorithm of~\citet{Badanidiyuru2014b}, it achieves a significant improvement in the number of queries at a small cost in the objective value.

\subsection{Related work}

\paragraph{Linear-Time Algorithms for Size and Knapsack Constraints.}
For the special case of monotone submodular maximization under a  size constraint,
two works \citep{Kuhnle2021a, Li2022} independently achieved an $\oh{(n / \epsi) \log( 1/ \epsi)}$ query complexity and a
nearly optimal $1-1/e-\epsi$ approximation ratio. These are the first deterministic algorithms to achieve nearly
the $1 - 1/e$ ratio for size constraints with linear query complexity.
\citet{Li2022} also provide an $\ohs{\epsi}{n \log n}$ algorithm for the intersection of matroid and
$\ell$ knapsack constraints; however, for the case of a single matroid and zero knapsacks,
this algorithm reduces to the $1/2 - \epsi$-approximation algorithm in \citep{Badanidiyuru2014b} previously discussed.

\paragraph{Relationship to \citet{Kuhnle2021a}.}
For \mon{} under size constraint,
\citet{Kuhnle2021a} provides the algorithm
\quickstream{}, a $(1/4)$-approximation algorithm
in exactly $n$ queries. Our algorithm shares
some common features with this algorithm: 
both algorithms query an infeasible set to determine
whether to add an element. However, \quickstream{} uses
multiple ideas that are specific to size constraint which
do not generalize to matroid constraints. Most importantly,
it relies upon the fact that the last $k$ elements added
to the infeasible set form a feasible set. This fact, together
with the condition to add an element, is crucial for proving
the $1/4$ ratio of \quickstream{}. Unfortunately,
for the matroid constraint, one cannot find
a feasible set in this way.
A different strategy for maintaining a feasible subset
of the infeasible set is required, together with a different
strategy for adding an element. 

We also note that for size constraint, \citet{Kuhnle2021a} gave a  technique to transform any $\alpha$-approximation deterministic algorithm with query complexity $q$ into an $\alpha/c$-approximation deterministic algorithm with query complexity $q/c$, for any integer $c$. However, this technique does not work  for matroid constraints. Another technique to reduce the query complexity at the cost of the approximation ratio is to run an algorithm on a randomly sampled subset of the ground set. However, the approximation guarantees with this approach only hold in expectation, instead of deterministically, and the sampling also causes a loss in approximation. In particular, for \nmon{}, this technique cannot be used with an existing algorithm to achieve a constant approximation and a linear query complexity.  

\paragraph{Relationship to \citet{Chakrabarti2015b}.}
As mentioned above, \citet{Chakrabarti2015b} developed a streaming algorithm
for the monotone problem that achieves the same ratio as our algorithm
in at most $2n$ queries. Their algorithm takes one pass through the ground set and
maintains a feasible set through the following swapping logic.
Each element is assigned a weight (which requires at most two
queries to the oracle), and
the feasible solution is updated via appealing to an
algorithm of \citet{Ashwinkumar2011} for maximum (modular) weight
independent set. By contrast, our monotone algorithm employs
a single query to an infeasible set to determine whether two elements should be swapped.
In our empirical evaluation (Section \ref{sec:exp}), we show that the two algorithms obtain a similar
objective value, but our algorithm uses  fewer queries. 

\paragraph{Relationship to \citet{Feldman2018a}.}
\citet{Feldman2018a} developed several streaming algorithms
for the general problem \nmon{}. These algorithms also take one
pass through the ground set and decide whether it makes sense to
swap out an existing element for a new candidate. These algorithms
employ many queries to the submodular function to determine if
a swap should be made: $O(k)$ queries are required per element in
the worst case, where $k$ is the rank of the matroid. By contrast,
our general algorithm makes two queries to two infeasible sets
to determine if an element should be swapped. However, we should
note that the graph constructions required to prove our approximation
ratios are inspired by the graph constructions used in the analysis
of these algorithms. 

\paragraph{Faster algorithms that achieve optimal ratio for \mon{}.}
Several works have improved the number of queries needed to obtain the optimal $1 - 1/e$ ratio
 for \mon{} \citep{Badanidiyuru2014b,Buchbinder2015,Kobayashi2024}.
Very recently, a ratio of
$1 - 1/e - \epsi$ was achieved  with $ \mathcal O (\sqrt{k} \, n \, \text{poly}(1/\epsi , \log n))$ queries \citep{Kobayashi2024}.

\subsection{Preliminaries}

A function $f : 2^\mathcal{N} \rightarrow \mathbb{R}_{\geq 0}$ is \textit{submodular} if it satisfies the following diminishing returns property: for all sets $S \subseteq T \subseteq \mathcal{N}$ and any element $a \in \mathcal{N} \setminus T$, we have $f(a|S) \geq f(a|T)$, where $f(a|S) = f(S + a) - f(S)$ is the marginal contribution of $a$ to $S$. Equivalently, $f$ is submodular if $f(S) + f(T) \geq f(S \cup T) + f(S \cap T)$ for all sets $S, T \subseteq \mathcal N$. It is \textit{monotone} if $f(S) \leq f(T)$ for all sets $S \subseteq T \subseteq \mathcal{N}$.

Let $\mathcal I \subseteq 2^\mathcal{N}$ be a collection of subsets of $\mathcal{N}$. Then $\mathcal M = (\mathcal N, \mathcal I)$ is a \textit{matroid} if the following three properties are satisfied: (1)
 $\emptyset \in \mathcal I$, \, (2)
 for all sets  $S \subseteq T \subseteq \mathcal{N}$, if $T \in \mathcal I$ then $S \in \mathcal I$ (downward-closed property), and \, (3) for all sets $S, T \subseteq \mathcal{N}$ such that $|S| < |T|$, there exists $e \in T \setminus S$ such that $S + e \in \mathcal{I}$ (augmentation property).

In the submodular maximization under a matroid constraint problem (\textit{\nmon{}}), we are given a submodular function $f : 2^\mathcal{N} \rightarrow \mathbb{R}_{\geq 0}$  and a matroid $\mathcal{M} = (\mathcal N, \mathcal I)$, and the goal is to approximately solve $\max_{S \in \mathcal I} f(S)$. In the monotone submodular maximization under a matroid constraint problem (\textit{\mon{}}), the function $f$ is assumed to also be monotone. A set $S \in \mathcal I$ is interchangeably called  independent or feasible. We assume without loss of generality that $\{e\} \in \mathcal I$ for all $e \in \mathcal N$.

The algorithm is given access to a value oracle for $f$, i.e., it can query the value $f(S)$ of any set $S \subseteq \mathcal N$, as well as an independence oracle for $\mathcal M$, i.e., it can test  whether $S \in \mathcal I$ or $S \not \in \mathcal I$ for any set $S \subseteq \mathcal N$.  In submodular maximization, the main bottleneck for the running time is typically the function evaluations $f(S)$, so we are interested in designing algorithms for \mon{} and \nmon{} with low \textit{query complexity}, which is the worst-case number of queries made by the algorithm to the value oracle for $f$.

\section{An approximation algorithm for \mon{} in exactly $n$ queries}

In this section, we present our $(1/4)$-approximation algorithm
for \mon{} that uses exactly $n$ queries.
Monotonicity is used in only one place in the analysis,
and we also invoke the lemmata proved here in Section \ref{sec:nm},
where we develop a constant-factor algorithm for \nmon{} with $2n$  query complexity. 

\paragraph{Description of the algorithm.} In overview, the algorithm
makes a single pass through the ground
set, and each element is swapped
into the solution $A'$ if it is good enough
relative to the element it displaces. The novelty
lies in the fact that each element is considered
to be swapped into $A'$ \textit{only once}; and
the definition of good enough relies upon a single
query of the marginal gain of the element to an
infeasible set $A \supseteq A'$ when it arrives.

Specifically, the algorithm maintains two sets, $A' \subseteq A$; $A'$ is
always a feasible solution, and $A$ contains all elements that
were once a member of $A'$. 
On Line \ref{line:delta-weight}, the weight $\delta_e$ of element $e$ is
defined as its marginal gain into $A$: $\delta_e = f( e | A )$.
Since $f(A)$ is already known, computing the gain requires a single query:
$f(A + e)$. 
The weight $\delta_e$ is fixed upon arrival of $e$ and never recomputed, and is used for
all processing related to $e$. Next, the best candidate
to swap $e$ into $A'$ is found, according to the weights
$\delta$. This element, denoted $a^*$, is chosen to be
a smallest weight element such
that $A' \setminus a^* + e$ is feasible; that is,
$a^* = \argmin_{a \in A'| A'\setminus a + e \in \ind} \delta_a$.
If $\delta_e$ is large enough relative to $\delta_{a^*}$,
then the swap occurs: $A'$ is updated to $A' \setminus a^* + e$,
and $A$ is updated to $A+ e$ (and the value of $f(A)$ is updated
to $f(A) + \delta_e$). 
\vspace{0.4cm}
 
\begin{algorithm}[H]
  \DontPrintSemicolon
  \textbf{Input:} Oracle to $f$, ground set $\uni$,
  independence oracle for $\mat = (\uni, \ind)$, parameter $\beta$ \; 
 $A, A' \gets \emptyset$ \;
 \For{$e \in \uni$}{
  $\delta_e \gets f(e | A )$ \label{line:delta-weight}\;
  \If{$A' + e \in \ind$ and $\delta_e \geq 0$}{
    $A \gets A + e$\;
    $A' \gets A' + e$\;
  }
  \Else{
  $a^* \gets \argmin_{a \in A'| A' -  a + e \in \ind} \delta_a$ \;
  \If{$\delta_e \geq (1 + \beta)\delta_{a^*}$}{
    $A \gets A + e$\;
    $A' \gets A' - a^* + e$
  }
  }
 }
 Return $A'$
 \caption{\quickmax{}: A $\frac{1}{4}$-approximation algorithm for monotone submodular maximization under a matroid constraint}
 \label{algo:original-matroid}
\end{algorithm}

\subsection{Overview of analysis}
The analysis proceeds by
first relating
$f(A)$ and $f(A')$ (Lemma \ref{lem:primehalf}),
then $\opt = f(O)$, $f(A)$, and $f(A')$ (Lemma \ref{lem:primeopt}).
We state these two lemmata, then prove the approximation ratio. 
Subsequently, we prove the lemmata in Sections \ref{lem:ph-proof} and \ref{lem:po-proof}, respectively.
Both lemmata hold for general, submodular functions, which
will be needed in Section \ref{sec:nm}. Finally, we show that
the $1/4$ ratio is tight, with a set of tight examples in Section \ref{sec:tight}.

Lemma \ref{lem:primehalf} relates $f(A')$ and $f(A)$; intuitively, because
$A' \subseteq A$ and $\delta_e = f(e | A)$, then by submodularity
and the condition to swap $a^*$ for $e$, the $f$-value of $A'$ increases by a
constant fraction of the increase in the $f$-value of $A$, despite the loss from $a^*$. 
\begin{restatable}{lemma}{primehalf}
\label{lem:primehalf}
Let $(f,\mat)$ be an instance of \nmon{},
and let $A',A$ be produced by Alg. \ref{algo:original-matroid} on
this instance.
Then
$f(A') \geq \frac{\beta f(A)}{1 + \beta}$.
\end{restatable}

Next, Lemma \ref{lem:primeopt} establishes a relationship
between $f(O \cup A)$, $f(A)$, and $f(A')$. Intuitively,
the rejected elements $O \setminus A$ can each be mapped
to an element of $A'$ responsible for the rejection. The
key fact, which much of the proof is devoted to showing,
is that this mapping is injective. That is, each element $o$
of $O \setminus A$ can be mapped to a unique element of $A'$,
which may be thought of as gatekeeper for the element $o$.
To prove the mapping is an injection, a graph construction is employed.
\begin{restatable}{lemma}{primeopt}
  \label{lem:primeopt}
  Let $(f,\mat)$ be an instance of \nmon{}
  with optimal solution $O$,
  and let $A',A$ be produced by Alg. \ref{algo:original-matroid}
  on this instance.
  Then $f(O \cup A) \le f(A) + (1 + \beta)f(A')$.
\end{restatable}
The approximation follows directly from these two lemmata, as summarized in the following theorem. 
\begin{theorem}
    Algorithm~\ref{algo:original-matroid} is a $(1/4)$-approximation algorithm for \mon{} with query complexity $n$.
\end{theorem}
\begin{proof}
  Let $e_i$, $A_i$, and $A'_i$ denote the element $e$, the set $A$, and the set $A'$ at iteration $i$ of the algorithm.
For the query complexity, note that  at iteration $i$, the algorithm evaluates $f(e_i|A_{i-1}) = f(A_{i-1} + e_i) - f(A_{i-1})$. Let $e_j$ be the last element added to $A$ such that $j < i$. If there is no such $j$, then $A_{i-1} = \emptyset$ and $f(A_{i-1}) = 0$. Otherwise, we have $A_{i-1} = A_j = A_{j-1} + e_j$ and query $f(A_{i-1}) = f(A_{j-1} + e_j)$ was already performed at iteration $j$. Thus, only one query to $f$ is needed at each iteration and the query complexity is $n$. For the approximation, observe that
\begin{align*}
f(O) \leq_{(1)} f(O \cup A) 
<_{(2)} f(A) + (1 + \beta)f(A') 
\leq_{(3)} \frac{(\beta + 1)^2}{\beta }f(A'), 
\end{align*}
where $(1)$ is by monotonicity, $(2)$ is by Lemma~\ref{lem:primeopt}, and $(3)$ is by Lemma~\ref{lem:primehalf}.
The $1/4$ ratio follows from optimizing over $\beta \in [0, \infty)$ (the ratio is optimized at $\beta = 1$).
\end{proof}

\subsection{Proof of Lemma~\ref{lem:primehalf}} \label{lem:ph-proof}

We first give a helper lemma, which relates the change in the sum of the
weights of elements in $A'$ to the same sum for $A$, for a single iteration.
\begin{restatable}{lemma}{half}
\label{lem:half} Let $f$ be a submodular function (not necessarily monotone). Then, for any $i \in [n],$ we have that 
    $\sum_{e_j \in A'_i } \delta_{e_j}- \sum_{e_j \in A'_{i-1} } \delta_{e_j} \geq \frac{\beta}{1 + \beta}\left(\sum_{e_j \in A_i } \delta_{e_j} - \sum_{e_j \in A_{i-1} } \delta_{e_j} \right).$
\end{restatable}
\begin{proof}
 There are three cases. Case 1: if $A'_{i-1} + e_i \in M$ and $\delta_{e_i} \geq 0$. Then, we have $A_i = A_{i-1} + e_i$ and $A'_i = A'_{i-1} + e_i.$ We get that
\begin{align*}
    \sum_{e_j \in A'_i } \delta_{e_j} - \sum_{e_j \in A'_{i-1} } \delta_{e_j}    = \delta_{e_i}  \geq_{(1)} \frac{\beta \delta_{e_i}}{1 + \beta}   
     = \frac{\beta}{1 + \beta}\left(\sum_{e_j \in A_i } \delta_{e_j} - \sum_{e_j \in A_{i-1} } \delta_{e_j} \right), 
\end{align*}
where $(1)$ is since $\delta_{e_i} \geq 0$.
Case 2: 
If  $A'_{i-1} + e_i \in M$ and $\delta_{e_i} \geq 0$ do not hold and $\delta_{e_i} \geq (1 + \beta) \delta_{a^{*}}$ does hold. This is the main case for this proof. In this case, we have $A_i = A_{i-1} + e_i$ and $A'_i = A'_{i-1} - a^{*}  + e_i.$ We get that
\begin{align*}
    \sum_{e_j \in A'_i } \delta_{e_j} - \sum_{e_j \in A'_{i-1} } \delta_{e_j}   
    =_{(1)} \delta_{e_i} -  \delta_{a^*}
     \geq_{(2)}  \frac{\beta}{1 + \beta} \delta_{e_i}    
     =_{(3)}   \frac{\beta}{1 + \beta}\left(\sum_{e_j \in A_i } \delta_{e_j} - \sum_{e_j \in A_{i-1} } \delta_{e_j} \right)  .
\end{align*}
where $(1)$ is since
$A'_i = A'_{i-1} - a^{*}  + e_i$, $(2)$ since
 $\delta_{e_i} \geq (1 + \beta ) \delta_{a^*}$, and $(3)$
since $A_i = A_{i-1} + e_i$. Case 3: if neither conditions hold for the two previous cases, we have that $A_i = A_{i-1}$ and $A'_i = A'_{i-1}$. This implies that $\sum_{e_j \in A'_i } \delta_{e_j} - \sum_{e_j \in A'_{i-1} } \delta_{e_j} = 0$ and $\sum_{e_j \in A_i } \delta_{e_j} - \sum_{e_j \in A_{i-1} } \delta_{e_j} = 0$, and we trivially obtain the desired claim.
\end{proof}

We are now ready to prove Lemma~\ref{lem:primehalf}.
\begin{proof}[Proof of Lemma~\ref{lem:primehalf}] 
 We first claim that $A'_n \cap \{e_1, \ldots, e_{j-1}\} \subseteq A_{j-1}$. Consider $e_i \in A'_n \cap \{e_1, \ldots, e_{j-1}\}$, so $i \leq j-1$.  Since $e_i \in A'_n,$ we have $e_i \in A'_i$ by definition of the algorithm. Since $i \leq j-1$, we have $A'_i \subseteq A_i \subseteq A_{j-1}$. Since $e_i \in A'_i$ and  $A'_i \subseteq A_{j-1}$, we get $e_i \in A_{j-1}$, which proves the desired claim. We get that
    \begin{align*}
        f(A'_n) & = \sum_{e_j \in A'_n} f(e_j | A'_n \cap \{e_1, \ldots, e_{j-1}\}) & \\ 
        & \geq \sum_{e_j \in A'_n} \delta_{e_j} &  A'_n \cap \{e_1, \ldots, e_{j-1}\}  \subseteq  A_{j-1} \\
        & = \sum_{i=1}^n \left( \sum_{e_j \in A'_i } \delta_{e_j} - \sum_{e_j \in A'_{i-1} } \delta_{e_j}\right) & \text{telescoping sum}\\
         & \geq  \sum_{i=1}^n\frac{\beta}{1 + \beta}\left(\sum_{e_j \in A_i } \delta_{e_j} - \sum_{e_j \in A_{i-1} } \delta_{e_j} \right) & \text{Lemma~\ref{lem:half}} \\
          & =\frac{\beta}{1 + \beta} \sum_{e_j \in A_n} \delta_{e_j} & \text{telescoping sum} \\
         & = \frac{\beta f(A_n)}{1 + \beta}. & \qedhere
    \end{align*}
\end{proof}

\subsection{Proof of Lemma~\ref{lem:primeopt}} \label{lem:po-proof}

The main lemma used to prove Lemma~\ref{lem:primeopt} is the following.

\begin{restatable}{lemma}{graphconstruction}
\label{lemma:graph-construction}
Let $f$ be a submodular function. For any independent set $S \subseteq \uni$,
there exists an injection $\phi: S \to A'$ such that 
$\delta_{\phi(a)} \ge \delta_{a}/(1 + \beta )$ for all $a \in S$.
\end{restatable}

 In order to prove this lemma, we use graph constructions inspired by \citet{Chekuri2015} and \citet{Feldman2018a}. At a high level, the proof will be divided into three parts as follows:
\begin{itemize}
    \item[1)] Construct a graph based on a run of the algorithm (\Cref{sec:graphc}).
    \item[2)] Verify particular properties of this graph (\Cref{sec:properties}).
    \item[3)] Invoke a technical lemma from \citet{Feldman2018a} which that allows  to use the graph's properties  to deduce \autoref{lemma:graph-construction} (\Cref{sec:usingproperties}). 
\end{itemize}
Finally, we use Lemma~\ref{lemma:graph-construction} to prove Lemma~\ref{lem:primeopt} in \Cref{sec:usinglemma}.

\subsubsection{Graph construction}
\label{sec:graphc}
  We initialize a directed graph $G_0$ with no edges and vertex set $V$ consisting of a vertex $v_i$
for each element $e_i$ of $\uni$.  Then, we obtain $G = G_n$ by iteratively constructing $G_i$ from $G_{i-1}$ via the addition of directed edges. The edges to be added are determined by the behavior of \process{} on iteration $i$, as follows. 
  If $A'_{i-1} + e_i$ is found to be independent and $f( e_i | A_{i - 1}) \ge 0$,
  so that $e_i$ is added to the current state of $A'$
without necessitating a swap, then no edges are added and $G_i = G_{i-1}$. Otherwise, we consider the set $U_i = \{e_i\} \cup \{a :  A'_{i-1} \setminus a + e_i \in I \}$ which contains $e_i$ and the potential elements $e_i$ could have been swapped with in this round.
Then, define $a^*_i$ to an $a \in U_i \setminus e_i$ with minimum $\delta_a$ value.
Observe that exactly one of $a^*_i$ or $e_i$ is in $A'_{i}$ (either $e_i$
had enough marginal contribution to cause $a^*_i$ to be swapped out, or it did not).
Denote $u_i$ to be the  element of $\{a^*_i, e_i\}$ that is \textit{not} in $A_i'$.
Then, $G_i = G_{i - 1} \cup \{ (u_i, a) : a \in U_i \setminus u_i \}$. That is,
we add a directed edge from $u_i$ to all the other elements of $U_i$. After $n$ iterations of this procedure, we obtain our final graph $G = G_n$. 

\subsubsection{Properties of the graph}
\label{sec:properties}
Next, we give properties that are satisfied by the graph $G$. We first note the following folklore theorem. 
\begin{theorem}[Folklore, as in \citet{unique_circuit}] \label{thm:folklore} Given some matroid $\mat = (\uni, \ind)$, if $S \in \ind$ but $S + e \notin \ind$ then there exists a unique circuit $C \subseteq S+e$.
\end{theorem}
From \autoref{thm:folklore}, we can show the following known proposition about circuits, whose proof we include for completeness.
\begin{restatable}{proposition}{swapcircuit}
  \label{prop:swap-circuit}
Given a matroid $\mat = (\uni, \ind)$, let $S \in \ind$ and $e \in \uni$, such that $S + e \not \in \ind$. Then $\{ a : S \setminus a + e \in \ind \} \cup \{ e \}$ is the unique circuit contained in $S + e$.
\end{restatable}
\begin{proof}
  Let $B = \{ a : S \setminus a + e \in \ind \} \cup \{ e \}$.
  First, $e \in C$ since otherwise $C \subseteq S$, and hence would be independent.
  So it suffices to show that $C \setminus e = B \setminus e$.

  Let $a \in B \setminus e$. Notice that since $S \setminus a + e \in \ind$,
  every subset of $S \setminus a + e$ is independent, and
  hence $C \not \subseteq S \setminus a + e$.  As $C$ is a subset of $S + e$, this implies that $a \in C$. Therefore,
  $B \setminus e \subseteq C \setminus e$.

  Next, let $a \in C \setminus e$. Then $C' = C \setminus a$ must be independent
by the inclusionwise minimality of $C$. While $|S| > |C'|$, we can iteratively add elements from $S$ to $C'$ while
preserving independence by the augmentation property of matroids. Observe that we never add $a$ to $C'$, as then $C \subseteq C'$
and $C'$ would not be independent. Hence, we finally obtain $C' = S \setminus a + e$.  Thus, $S \setminus a + e \in \ind$,
and $a \in B \setminus e$. Therefore, $C \setminus e \subseteq B \setminus e$. 
\end{proof}

We are now ready prove the graph properties.

\begin{property}  All non-sinks $v$ of $G$ are spanned by the set $\delta^+(v) = \{ x : (v, x) \in G \}$.
 \label{prop:non-sink}
\end{property}
\begin{proof}
  From the graph construction, out-edges are added to $v$ if and only if $v = u_i$ at some
  iteration $i$. This means that $e_i + A'_{i-1}$ is not independent, although $A'_{i-1}$ is independent.
  Therefore, by \autoref{prop:swap-circuit}, $U_i = \{e_i\} \cup \{x :  A'_{i-1} \setminus x + e_i \in \ind \}$
  is a circuit, and by construction each $u \in U_i \setminus u_i$ is an out-neighbor of $v$. Moreover,
  rank$(U_i \setminus u_i)$ = rank($U_i$), so $U_i \setminus u_i$ spans $u_i$. Since $\delta^+(v) \supseteq U_i \setminus v$,
  it holds that $\delta^+(v)$ also spans $v$.
\end{proof} 
Next, by simply inspecting the behavior of the algorithm in how it chooses $u_i$ in each iteration $i$, we get the following two properties.
\begin{property} \label{prop:obv}
  No element $e_k$ can be designated as both $u_i, u_j$ for some $i \neq j$ during the construction of $G$.
\end{property}
\begin{proof}
  Once an element $e_k$ is designated $u_i$ at some iteration $i$, by construction $e_k \not \in A'_i$, and $k \le i$.
  Moreover, the only possible candidates to be added to $A'_i$ are $\{e_{i+1}, \ldots, e_n\}$. So $e_k \not \in A'_j$,
  for any $j \ge i$ and hence cannot be chosen as $u_j$ for some iteration $u > i$. 
\end{proof} 
\begin{property} \label{prop:sink}
  An element $e_i \in A'$ implies its corresponding vertex in $G$ is a sink.
  Conversely, a vertex $v$ in $G$ is a sink implies that it is in $A'$.
\end{property}
\begin{proof}
  An element is chosen as $u_i$ in some iteration $i$ iff out-edges are added to its corresponding vertex in iteration $i$ iff $u_i$ is not a sink.
  By the proof of \autoref{prop:obv}, 
  an element $s$ chosen as $u_i$ in some iteration implies that $s \not \in A'$. Further, if an element $s$ was never chosen
  as $u_i$ in any iteration $i$, necessarily $s \in A'$.
\end{proof} 
This leads very naturally to the following desirable property of $G$.

\begin{property} $G$ does not contain any cycles. \label{G:acyclic} \end{property}
\begin{proof}
  To see this is the case, simply note that out-edges are only ever added to vertices which are at some point designated as $u_i$ for some iteration $i$ in the construction of $G$.  Hence, if $G$ had a cycle, necessarily this fact and Property~\ref{prop:obv} above would imply that it would have to contain some edge $(u_j, u_i)$ for $j > i$ where $u_i \neq u_j$.  However, by inspection this clearly yields a contradiction, as $u_i$ must have been such that it was not in $A'_i$, and hence clearly not in $A'_j$, so $u_j$ could never have an outedge from itself to $u_i$ (as all neighbors of $u_j$ must be in $A'_j$).  Thus, $G$ cannot contain a cycle. \end{proof} 
Finally, we state and prove our last property:
\begin{property} Let element $a$ be reachable in $G$ from element $e$. Then $\delta_e \le (1 + \beta)\delta_a$. \label{prop:delta}
\end{property}
\begin{proof}
  Let $(y,x)$ be any edge in $G$. Observe that all edges added during the graph construction
  satisfy that the target vertex is in $A$. If also $y \in A$, it means that $y = u_j = a_j^*$
  on some (unique) iteration $j$; and hence by the selection of $a_j^*$,
  $\delta_y \le \delta_x$. If it is the case that $y \not \in A$,
  it means $y = u_j = e_j$ on some iteration $j$; since $y$ is rejected, it must hold that
  $\delta_y < (1 + \beta ) \delta_x$.

  Consider the edges on a path from $e$ to $a$: $(e, x_1),(x_1, x_2),\ldots,(x_{k-1}, x_k = a)$. By the
  above observation $\{ x_1,\ldots,x_k \} \subseteq A$,.
  so $\delta_{x_i} \le \delta_{x_{i+1}}$ for all $i \in 1, \ldots, k-1$; also, $\delta_e \le (1 + \beta )\delta_{x_1}$.
  Therefore, $\delta_e \le (1 + \beta )\delta_{a}$.
\end{proof}

\subsubsection{Using the graph properties}
\label{sec:usingproperties}

  We introduce a technical lemma from \citet{Feldman2018a}.
\begin{lemma}[Lemma 13 in \citep{Feldman2018a}] \label{lemm:machinery} Consider an arbitrary directed acyclic graph $G$ whose vertices are elements of some matroid $\mat$.  If every non-sink vertex $v$ of $G$ is spanned by $\delta^+(v)$ in $\mat$, then for every set $S$ of vertices of $G$ which is independent in $\mat$ there must exist an injective function $\phi_S$ such that, for every vertex $v \in S$, $\phi_S(u)$ is a sink of $G$ which is reachable from $u$.
\end{lemma}
We are ready to formally prove Lemma \ref{lemma:graph-construction}, which we restate for convenience.
\graphconstruction*
\begin{proof}
  Consider the graph $G$ which we constructed above. By Lemma \ref{lemm:machinery}, Property~\ref{prop:non-sink},
  and  Property~\ref{G:acyclic} guarantee the existence, for any
  independent set $S \subseteq Y$, of an injective function $\phi$ mapping elements of $S$ to a set
  $T$ of sinks of $G$, such that $\phi(s)$ is reachable from $s$. By Property~\ref{prop:sink},
  $T \subseteq A'$. Finally, by Property~\ref{prop:delta}, it follows that $\delta_{s} \le (1 + \beta )\delta_{\phi(s)}$.
\end{proof}

\subsubsection{Using Lemma~\ref{lemma:graph-construction} to prove Lemma~\ref{lem:primeopt}}
\label{sec:usinglemma}

Finally, we use Lemma~\ref{lemma:graph-construction} to prove Lemma~\ref{lem:primeopt}.

\primeopt*
\begin{proof}
By Lemma~\ref{lemma:graph-construction}, there exists an injection $\phi : O \rightarrow A'$ such that $\delta_{\phi(e_j)} \ge \delta_{e_j} / (1 + \beta )$ for all $e_j \in O$. We get that 
\begin{align*}
    f(O \cup A)  - f(A) & \leq \sum_{e_j \in O \setminus A} \delta_{e_j} & \text{submodularity} \\
    &\le  \sum_{e_j \in O \setminus A} (1 + \beta) \delta_{\phi(e_j)}  & \\
      & =    \sum_{\substack{e_i \in A':  \ e_i = \phi(e_j)  \\ \text{for some } e_j \in O \setminus A}} (1 + \beta ) \delta_{e_i} & \text{$\phi(e_i) \in A'$ for $e_j \in O$}\\
            &  \leq    \sum_{e_i \in A'_n} (1 + \beta )  \delta_{e_i} & \text{$\delta_{e_i} \geq 0$ for  $e_i \in A_i$} \\
     &   \leq   \sum_{e_i \in A'_n} (1 + \beta ) f(e_i | A'_{i-1}) & \text{submodularity}\\
      &   =  (1 + \beta )f(A'_n). & \qedhere
\end{align*}
\end{proof}

\subsection{Tight examples} \label{sec:tight}

In this section, we describe a set of instances  for which \quickmax{} gets a ratio arbitrarily close to $1/4$, which shows that the analysis of the preceding sections is tight. Let $\epsi > 0$. We construct an instance where the set $A'$ returned by \quickmax{} is
less $\opt / (4 - \epsi)$. Let $m$ be an integer greater than $\log_2( 1 / \epsi ) + 1$,
and let $\uni = \{ x_0, x_1, \ldots, x_m, x_{m+1} = o \}$ be an ordered set of $m + 2$ elements.
For each $i < m + 1$, let $g(x_i) = 2^i$. Let $g(o) = 2^{m+2} - 2$. For any $S \subseteq \uni$,
define $g(S) = \sum_{s \in S} g(s)$; thus, $g$ is a modular fuction. Finally, define
$f(S) = \min \{ g(S), g(o) \}$; then $f$ is a monotone, submodular function. Then, consider
a size constraint of $k = 1$. The optimal solution on this instance is clearly $\{ o \}$.

Consider the run of \quickmax{} on this instance, where the elements of $\uni$ are processed
in the given ordering (for convenience, number the iterations of the \textbf{for} loop
from $0$). Suppose inductively at iteration $i - 1$, $A_{i-1} = \{x_0, \ldots, x_{i-1}\}$,
and $A'_{i-1} = \{x_{i-1} \}$; this is satisfied at iteration $1$ since $\{ x_0 \}$ is feasible at
iteration $0$. 
Then, at iteration $i < m + 1$,
\begin{align*}
  g(A_{i - 1} + x_i) = \sum_{j=0}^i g(x_j) = \sum_{j=0}^i 2^j = 2^{i+1} - 1 < g(o).
\end{align*}
Therefore, $f(A_{i - 1} + x_i) = g(A_{i-1} + x_i)$ and thus $\delta_{x_i} = g( x_i ) = 2^i$.
Since inductively, $A'_{i - 1} = \{ x_{i -1 } \}$, and $\delta_{x_i} \ge 2 \delta_{x_{i-1}}$,
a swap is made: $A'_i = A'_{i - 1} \setminus x_{ i - 1} + x_i$ and $A_i = A_{i - 1} + x_i$,
which was to be shown.

Now, consider iteration $m + 1$, by the above argument $A_m = \{ x_1, \ldots, x_m \}$ and $A'_m = \{ x_m \}$.
By definition $f(A_m + o) = g(o) = 2^{m + 2} - 2$. Hence
$$\Delta ( o | A_m ) = 2^{m + 2} - 2 - (2^{m + 1} - 1) = 2^{m +1} - 1 < 2\delta_{x_m} = 2^{m + 1}.$$
Thus $o$ is rejected, and the algorithm terminates with $A' = \{ x_m \}$. Moreover,
$$ \frac{f(o)}{ f(x_m) } = \frac{2^{m +2} - 2}{2^m} = 4 - 2^{1 - m} > 4 - \epsi,$$
  since $m \ge \log_2(1 / \epsi ) +1$. 

\section{Approximation algorithm for \nmon{} with linear query complexity} \label{sec:nm}

In this section, we present the first constant-factor algorithm
with linear query complexity for general submodular
objectives under a matroid constraint. Specifically,
Alg. \ref{algo:non-monotone} achieves ratio
$\approx 1/11.67$ with exactly $2n$ queries to $f$.  First, we discuss why our algorithm \quickmax{} does not achieve a ratio for general, submodular objectives;
Lemmata \ref{lem:primehalf},\ref{lem:primeopt} establish a relationship between $f( O \cup A )$ and $f(A')$,
where $O$ is an optimal solution. 
Since $f$ is non-monotone, it may hold that $f( O \cup A )$ is smaller
than $f(O)$ and may have no non-trivial lower bound.

\paragraph{Description of the algorithm.}
To deal with this challenge, at a high level, we run two copies of \quickmax{}
concurrently, making sure all of the sets maintained are disjoint
between the two versions. The first copy maintains sets $A' \subseteq A$ as before,
and the second copy maintains sets $B' \subseteq B$. To ensure the sets $A \cap B$
are disjoint, an element $e$ is processed only by the copy that would assign
a larger weight $\delta_e$ to the element; that is, the copy
that determines a larger marginal gain to its infeasible set ($A$ or $B$).
Making this determination requires two queries to $f$, which are the only
queries required for processing the element.

\begin{algorithm}[ht]
\DontPrintSemicolon
\textbf{Input:} function $f$, matroid $M$, element $e$, sets  $S$, set $S'$, contributions $\{\delta_a\}_{a \in S'}, \delta_e$, parameter $\beta$ \;
    \If {$S' + e \in M$ and $\delta_e \geq 0$}{
    $S \gets S + e$\;
    $S' \gets S' + e$
  }
  \Else{
  $a^* \gets \argmin_{a \in \{ S': S' - a + e \in M \} } \delta_a$ \;
  \If {$\delta_e \geq (1+\beta) \; \delta_{a^*}$}{
    $S \gets S + e$\;
    $S' \gets S' - a^* + e$
  }
}
\Return{ $S, S'$ }
 \caption{The \process{} subroutine}
 \label{algo:process}
\end{algorithm}

\begin{algorithm}[ht]
\DontPrintSemicolon
\textbf{Input:} function $f$, matroid $M$, ground set $N = \{e_1,\ldots,e_n\}$, parameter $\beta$  \;
 $A, A', B, B' \gets \emptyset$ \;
 \For{$i = 1$ to $n$}{
  \If{$f(e | A) > f(e | B)$}{
    $\delta_e \gets f(e|A)$ \;
    $A,A' \gets$ \process{}$(f, M, e, A, A', \{\delta_a\}_{a \in A'}, \delta_e, \beta)$\;
    $B,B' \gets B,B'$
  }
  \Else{
  $\delta_e \gets f(e|B)$ \;
    $B, B' \gets $ \process{}$(f, M, e, B, B',\{\delta_a\}_{a \in B'}, \delta_e, \beta)$\;
    $A,A' \gets A,A'$
  }
 }
 Return $\argmax\{f(A'), f(B')\}$
 \caption{\algnm{}}
 \label{algo:non-monotone}
\end{algorithm}


\paragraph{The analysis.} Let $N_A = \{e : f(e | A_{i-1}) > f(e | B_{i-1})\}$ and $N_B = N \setminus N_A$.  Consider  $e_i \in N$, then Algorithm~\ref{algo:non-monotone} either calls  \textsc{Process} over $A$ or over $B$ for $e_i$. In the first case, we say that $e_i$ is processed by $A$ and in the other we say that it is processed by $B$. Let $O_A$ and $O_B$ be the optimal elements $O$ that are processed by $A$ and $B$, respectively.  The main observation that allows utilizing parts of the analysis of the monotone algorithm for the analysis of the above algorithm for non-monotone functions is that Algorithm~\ref{algo:non-monotone} is equivalent to running Algorithm~\ref{algo:original-matroid} twice, once over $N_A$ and once over $N_B$, to obtain $A'$ and $B'$ respectively. However, note that $N_A$ and $N_B$ are not initially known and that whether $e_i \in N_A$ or $e_i \in N_B$ crucially depends on $A_{i-1}$ and $B_{i-1}$, which is why we cannot simply call Algorithm~\ref{algo:original-matroid} over $N_A$ and $N_B$.

\begin{lemma}
    Let $A'_{nm}$ and $B'_{nm}$ be the sets $A'$ and $B'$ computed by Algorithm~\ref{algo:non-monotone} over an arbitrary function $f$, matroid $\mathcal M$. Let $A'_{m}$ and $B'_{m}$ be the sets computed by Algorithm~\ref{algo:original-matroid} over $f, \mathcal M, N_A$ and $f, \mathcal M, N_B$, respectively. Then, we have that $A'_{nm} = A'_m$ and $B'_{nm} = B'_m$
\end{lemma}
\begin{proof}
    Observe that the \textsc{Process} subroutine is identical to Lines~5-12 of Algorithm~\ref{algo:original-matroid}.
\end{proof}

Since all the previous lemmas  hold for non-monotone functions (monotonicity was only used in the proof of the main theorem for monotone functions), the above lemma implies that previous lemmas  apply to $A'$ and $B'$ over ground sets $N_A$ and $N_B$. The only new lemma needed is the following.

\begin{restatable}{lemma}{nmmain}
\label{lem:nmmain}
    For any submodular function $f$, consider $A_n, B_n, A'_n, B'_n$ from Algorithm~\ref{algo:non-monotone}, we have 
    $\max\{f(O \cup B_n) - f(B_n), f(O \cup A_n) - f(A_n)\} \leq (1+\beta)(f(B'_n) + f(A'_n)).$
\end{restatable}
\begin{proof}
 Since $O = O_A \cup O_B$, we get that
    \begin{align*}
        f(O \cup A_n) - f(A_n)  = f(O_B \cup O_A \cup A_n) - f(O_A \cup A_n) + f(O_A \cup A_n)  - f(A_n).
    \end{align*}
    Next, we have 
    \begin{align*}
        & f(O_B \cup O_A \cup A_n) - f(O_A \cup A_n)\\
         \leq & \sum_{e_j \in O_B} f(e_j | O_A \cup A_n) & \text{submodularity} \\
         \leq & \sum_{e_j \in O_B} f(e_j | A_{i-1}) & \text{submodularity,} A_i \subseteq O_A \cup A_n, \\
       \leq  &  \sum_{e_j \in O_B} f(e_j | B_{i-1}) & \text{definition of $O_B$ and Algorithm~\ref{algo:non-monotone}} \\
       <  &  \sum_{e_j \in O_B} (1+\beta) f( e_{\phi_{O_B}(e_j)} | B_{\phi_{O_B}(e_j)-1}) & \text{Lemma~\ref{lemma:graph-construction}} \\
       \leq  &  \sum_{\substack{e_i \in B'_n: \\  i = \phi_{O_B}(e_j)  \\ \text{for some } e_j \in O_B}} (1+\beta) f(e_i | B_{i-1}) & \text{$e_{\phi_{O_B}(e_i)} \in B'_n$ for $e_i \in O_B$}\\
              \leq  &  \sum_{e_i \in B'_n} (1+\beta) f(e_i | B_{i-1}) & \text{$f(e_i | B_{i-1}) \geq 0$ for  $e_i \in B_i$} \\
       \leq  &  \sum_{e_i \in B'_n} (1+\beta) f(e_i | B'_{i-1}) & \text{submodularity}\\
       =  &  (1+\beta)f(B'_n) &
    \end{align*}
    
It is important to note that to apply submodularity for the second inequality, we have that $e_j \not \in O_A \cup A_n \text{ for } e_j \in O_B$. This holds since elements in $O_A \cup A_n$ are processed by $A$ and elements in $O_B$ are processed by $B$. We also note that $f(e_i | B_{i-1}) \geq 0$ for all $e_i \in B_i$ is by definition of  \textsc{Process} subroutine.

Next, observe that by submodularity we have
\begin{align*}
    f(O_A \cup A_n)  - f(A_n) \leq \sum_{e_j \in O_A \setminus A_n} f(e_j |A_n) \leq \sum_{e_j \in O_A \setminus A_n} f(e_j |A_{i-1})
\end{align*}
We get, similarly as above,
\begin{align*}
    &\sum_{e_j \in O_A \setminus A_n} f(e_j |A_{i-1})& \\
    < & \sum_{e_i \in O_A \setminus A_n} (1+\beta) f( e_{\phi_{O_A}(e_i)} | A_{\phi_{O_A}(e_i)-1}) &\text{Lemma~\ref{lemma:graph-construction}} \\
       \leq  &  \sum_{\substack{e_i \in A'_n: \\  i = \phi_{O_A}(e_j)  \\ \text{for some } e_j \in O_A \setminus A_n}} (1+\beta) f(e_i | A_{i-1}) & \text{$e_{\phi_{O_A}(e_i)} \in A'_n$ for $e_i \in O_A$}\\
              \leq  &  \sum_{e_i \in A'_n} (1+\beta) f(e_i | A_{i-1}) & \text{$f(e_i | A_{i-1}) \geq 0$ for  $e_i \in A_i$} \\
       \leq  &  \sum_{e_i \in A'_n} (1+\beta) f(e_i | A'_{i-1}) & \text{submodularity}\\
       =  &  (1+\beta)f(A'_n) &
\end{align*}
By combining the four previous series inequalities, we obtain that
  \begin{align*}
        f(O \cup A_n) - f(A_n) &   = f(O_B \cup O_A \cup A_n) - f(O_A \cup A_n) + f(O_A \cup A_n)  - f(A_n) \\
        & <  (1+\beta)f(B'_n) +  (1+\beta)f(A'_n).
    \end{align*}
We also have that $ f(O \cup B_n) - f(B_n)    <  (1+\beta)f(B'_n) +  (1+\beta)f(A'_n),$ which follows identically as for the bound on $f(O \cup A_n) - f(A_n)$. We conclude that $\max\{f(O \cup B_n) - f(B_n), f(O \cup A_n) - f(A_n)\} \leq (1+\beta)(f(B'_n) + f(A'_n)).$
\end{proof}

We are now ready to prove the main result for non-monotone functions.

\begin{theorem} For \nmon, Algorithm~\ref{algo:non-monotone} has query complexity $2n$ and achieves a $1/(6 + 4\sqrt{2}) \approx 1/ 11.66$ approximation.
\end{theorem}
\begin{proof}
At iteration $i$, the algorithm evaluates $f(e_i|A_{i-1}) = f(A_{i-1} + e_i) - f(A_{i-1})$ and $f(e_i|B_{i-1}) = f(B_{i-1} + e_i) - f(B_{i-1})$. Since queries $f(A_{i-1})$ and $f(B_{i-1})$ have already been evaluated in a previous iteration, the algorithm performs two queries at iteration $i$, $f(A_{i-1} + e_i)$ and $f(e_i|B_{i-1})$, and the total number of queries is thus $2n$. For the approximation, we have that
    \begin{align*}
        f(O) & =  f(O \cup (A_n \cap  B_n)) &  A_n \cap B_n = \emptyset \\
        &\leq  f(O \cup A_n \cup B_n) + f(O \cup (A_n \cap  B_n)) & \text{non-negativity}\\
        & \leq f(O \cup A_n) + f(O \cup B_n)  & \text{submodularity} \\
       & \leq  2(1+\beta)(f(B'_n) + f(A'_n)) + f(A_n) + f(B_n) &\text{Lemma~\ref{lem:nmmain}}\\
        & \leq  2(1+\beta)(f(B'_n) + f(A'_n) )+  (1+\beta)f(A'_n)/\beta +  (1+\beta)f(B'_n)/\beta &\text{Lemma~\ref{lem:primehalf}}\\
        & \leq 2(2(1+\beta) + (1+\beta)/\beta)  \max\{f(A_n'), f(B_n')\}.
    \end{align*}
    Finally, $2(2(1+\beta) + \frac{1+\beta}{\beta})$ is minimized at $\beta = 1/\sqrt{2}$, where $2(2(1+\beta) + \frac{1+\beta}{\beta}) = 6 + 4 \sqrt{2}.$
  \end{proof}

\section{Approximation algorithm for \pmat{} with linear query complexity} \label{sec:pmat}
\paragraph{Description of the algorithm.}
In overview, The Alg.\ref{alg:quick_p_matchoid} is the extension of Alg.\ref{algo:original-matroid} from the matroid constraint to the $p$-matchoid constraint and keeps the key feature $n$ queries unchanged. Compared with Alg.\ref{algo:original-matroid}, Alg.\ref{alg:quick_p_matchoid} should use Alg.\ref{alg:exchange_p_matchoid} as a subroutine to select candidates for exchange.

\paragraph{Relationship to \citet{Feldman2018a}.} Our subroutine Alg.\ref{alg:exchange_p_matchoid}: \excha{} is similar to the Alg.1 in \citet{Feldman2018a}. The difference is that we don't have any function queries in the \excha{}, instead, we use the previous query value $\delta_x$ of element $x$ for calculation without re-queries to return the exchanged candidates set $C$, which helps us keep exactly one query per element. Besides, we keep the same ratio $1/(4p)$ compared with \citet{Feldman2018a}, but ours is a deterministic algorithm.

\begin{algorithm}[ht]
\DontPrintSemicolon
\textbf{Input:} Oracle $f$, ground set $\mathcal{N}$, $p$-matchoid $ \mathcal{M} = (\mathcal{N}, \mathcal{I})$, parameter $\beta$ \;
$A, A' \gets \varnothing$ \;
\For{$e \in \mathcal{N} $}{
        Let $\delta_e \gets f(e \mid A)$  
        \label{line:query_p_matchoid} \;
        Let $C \gets $ \excha{}$(\mathcal{M}, A', e, \{\delta_x\}_{x \in A'})$ \;
        \If{$\delta_e \geq (1 + \beta) \cdot \sum_{a^* \in C} \delta_{a^*}$ \label{line:exchange_p_mat} }{
            $A \gets A  + e$ \;
            $A' \gets A' \setminus C + e$ \label{line:update_Aprime_p_mathchoid} \;
        }
    }
\textbf{return} $A'$
\caption{\quickpm{}: A $1 / (4p)$-approximation algorithm for monotone submodular maximization under a $p$-matchoid constraint}
\label{alg:quick_p_matchoid}
\end{algorithm}

\begin{algorithm}[ht]
\DontPrintSemicolon
\textbf{Input:} $p$-matchoid $ \mathcal{M} = (\mathcal{N}, \mathcal{I})$, set $A'$, element $e$, contributions $\{\delta_x\}_{x \in A'}$  \;
$C \gets \varnothing$ \;
\For{$\ell = 1$ to $m$}{
        \If{$(A' + e) \cap \mathcal{N}_\ell \notin \mathcal{I}_\ell$}{
            Let $X_\ell \gets \{x \in A' \mid ((A' - x + e) \cap \mathcal{N}_\ell) \in \mathcal{I}_\ell\}$ \;
            Let $x_\ell \gets \arg\min_{x \in X_\ell} \delta_x$ \;
            $C \gets C + x_\ell$. \;
            }
        }
\textbf{return} $C$ 
\caption{The \excha{} subroutine}
\label{alg:exchange_p_matchoid}
\end{algorithm}

\subsection{Overview of analysis}
The analysis proceeds by first relating $\delta( \cdot )$ with $f(\cdot)$, which are presented in Proposition \ref{prop:Aprime_delta_p_matchoid} and \ref{prop:A_delta_p_matchoid}. Then, we could build a connection between $\Delta_{i+1} \delta(A')$ and $\Delta_{i+1} \delta(A)$, which is presented in Lemma \ref{lemm:delta_p_matchoid}. By these relationships, we could bound $f(A')$ by $f(A\cup O)$. Considering the monotonicity of $f$, we could get Theorem \ref{theo:quick_p_matchoid}. At first, we have two propositions about the original function $f(\cdot )$ and the helper function $\delta( \cdot )$. 

\begin{proposition}  \label{prop:Aprime_delta_p_matchoid} We have that
    $f(A') \geq \delta(A').$
\end{proposition}

\begin{proof} Observe that
    \begin{align*}
        f(A') &= \sum^{|A'|}_{j=1} f(e_j \mid  \{e_1, \dots, e_{j-1} \})  
                = \sum^{|A'|}_{j=1} f(e_j \mid A'_{j-1}-1)  \\
                &\geq \sum^{|A'|}_{j=1} f(e_j \mid A_{i(e_j)} ) \tag{By Submodularity and $  A'_{j-1} \subseteq A_{i(e_j)-1} $} \\
                &= \sum_{e_j \in A'} \delta(e_j) 
                = \delta(A') \qedhere
    \end{align*}
\end{proof}

\begin{proposition}  \label{prop:A_delta_p_matchoid} We have that
    $f(A) = \delta(A).$
\end{proposition}

\begin{proof} Observe that
    \begin{align*}
        f(A) &= \sum^{|A|}_{i=1} f(e_i \mid  \{e_1, \dots, e_{i-1} \})  
                = \sum^{|A|}_{i=1} f(e_i \mid A_{i-1} ) \\
                &= \sum_{e_i \in A} \delta(e_i)  
                = \delta(A) \qedhere
    \end{align*}
\end{proof}

\begin{lemma} \label{lemm:A_minus_Aprime_p_matchoid}
    Let $(f, \mathcal{M})$ be an instance of \pmat, and let $A'$, $A$ be produced by Alg.\ref{alg:quick_p_matchoid} on this instance. let $\delta(\cdot)$ be a function $ \sum f(\cdot \mid A)$. Then
    $$
    \delta(A \setminus A') \leq \frac{1}{\beta} \cdot \delta(A'). 
    $$
\end{lemma}

\begin{proof}
    At first, consider the increment of the helper function $\delta(A')$, we have:
    \begin{align*}
        \Delta_{i+1} \delta(A') &= \sum_{e_i \in A'_i} \delta_{e_i} - \sum_{e_i \in A'_{i-1}} \delta_{e_i}  \\
                            &= \delta_e -  \sum_{a^* \in C_{i}} \delta_{a^*}  \tag{By Line \ref{line:update_Aprime_p_mathchoid} in Alg.\ref{alg:quick_p_matchoid}} \\
                            &\geq \beta \cdot \sum_{a^* \in C_{i}} \delta_{a^*}  \tag{By Line \ref{line:exchange_p_mat} in Alg.\ref{alg:quick_p_matchoid}}    \\   
                            &= \beta \cdot \delta(C_i)
    \end{align*}

    Now, we observe that every element of $A\setminus A'$ must have been removed exactly once from the final solution $A'$, which implies that $\{ C_i \mid e_i \in A \}$ are disjoint partitions of $A\setminus A'$. Using this observation, we get:
    
    \begin{align*}
        \delta(A \setminus A') &= \sum_{e_i \in A} \sum_{a^* \in C_{i}} \delta_{a^*}   
                            = \sum_{e_i \in A}   \delta(C_i) \\
                            &\leq \sum_{e_i \in A}  \frac{1}{\beta} \cdot \Delta_{i+1} \delta(A')  \\
                            &= \frac{1}{\beta} \cdot (\delta(A') - \delta(\varnothing)) \tag{By Telescoping Sum}  \\
                            &\leq \frac{1}{\beta} \cdot \delta(A') \tag{By Non-negativity of $f$}
    \end{align*}
\end{proof}

\begin{lemma}  \label{lemm:delta_p_matchoid}
    Lef $f$ be a submodular function (not necessarily monotone). Then, for $i \in [n]$, we have $\Delta_{i+1} \delta(A') \geq \frac{\beta}{1+\beta} \cdot \Delta_{i+1} \delta(A) $.
\end{lemma}

\begin{proof}
    \begin{align*}
        \Delta_{i+1} \delta(A') &= \sum_{e_j \in A'_i} \delta_{e_j} - \sum_{e_j \in A'_{i-1}} \delta_{e_j}  \\
                            &= \delta_e -  \sum_{a^* \in C_{i}} \delta_{a^*}  \tag{By Line \ref{line:update_Aprime_p_mathchoid} in Alg.\ref{alg:quick_p_matchoid}} \\
                            &\geq \frac{\beta}{1 + \beta} \delta_e   \tag{By Line \ref{line:exchange_p_mat} in Alg.\ref{alg:quick_p_matchoid}}    \\                   
                            &=   \frac{\beta}{1+\beta} \left(   \sum_{e_j \in A_i} \delta_{e_j} - \sum_{e_j \in A_{i-1}} \delta_{e_j}  \right)  \\
                            &= \frac{\beta}{1 + \beta} \cdot \Delta_{i+1} \delta(A) 
    \end{align*}
\end{proof}

\begin{lemma}   \label{lemm:Aprime_A_p_matchoid}
    Let $(f, \mathcal{M})$ be an instance of \pmat, and let $A'$, $A$ be produced by Alg.\ref{alg:quick_p_matchoid} on this instance. Then
    $$f(A') \leq \frac{\beta}{\beta + 1} \cdot f(A).$$
\end{lemma}

\begin{proof}
    \begin{align*}
        f(A'_n) &\geq \delta(A'_n)  \tag{By Proposition \ref{prop:Aprime_delta_p_matchoid}} \\
                &= \sum^n_{i=1} \Delta_{i+1} \delta(A')  \tag{By Telescoping Sum}  \\
                &\geq  \sum^n_{i=1} \frac{\beta}{1 + \beta} \cdot \Delta_{i+1} \delta(A)  \tag{By Lemma \ref{lemm:delta_p_matchoid}}  \\
                &=  \frac{\beta}{1 + \beta}  \sum_{e_j \in A_n} \delta_{e_j} 
                = \frac{\beta}{1 + \beta}  \cdot  \delta(A_n) \tag{By Telescoping Sum} \\
                &= \frac{\beta}{1 + \beta}  \cdot  f(A_n) \tag{By Proposition \ref{prop:A_delta_p_matchoid}}
    \end{align*}
\end{proof}

\begin{reduction}  \label{redu:quick_p_matchoid}
    For the sake of analyzing the approximation ratio of Alg.\ref{alg:quick_p_matchoid}, one may assume that every element $e \in \mathcal{N}$ belongs to exactly $p$ out of the $m$ ground sets $\mathcal{N}_1, \dots, \mathcal{N}_m$ of the matroids defining $\mathcal{M}$.
\end{reduction}

\begin{proof}
    This reduction is trivial, we could add the element $e$ to $(p-p')$ additional matroids as a free element, whose addition will keep the independence of matroids. These additions will not affect the behavior of Alg.\ref{alg:exchange_p_matchoid}.
\end{proof}

By Reduction \ref{redu:quick_p_matchoid}, we have a proposition to prove the approximation ratio (analagous
to Proposition 10 of \citet{Feldman2018a}).
\begin{property} \label{prop:exchange_p_matchoid}
    For every set $T \in \mathcal{I}$, there exists a mapping $\phi_T$ from elements of $T$ to multi-subsets of $A_n$ such that:
    \begin{enumerate}
        \item every element $e \in A'_n$ appears at most $p$ times in the multi-sets of $\{ \phi_T(e) \mid  e \in T \} $.
        \item every element $e \in A_n \setminus A'_n$ appears at most $(p-1)$ times in the multi-sets of $\{ \phi_T(e) \mid  e \in T \} $.
        \item every element $e \in T \setminus A_n$ obeys $ \delta_{e} \leq (1+\beta) \cdot \sum_{a^* \in \phi_T(e)} \delta_{a^*} $. 
        \item every element $e \in T \cap A$ obeys $\delta_{e} \leq \delta_{a^*}, \forall a^* \in \phi_T(e)$ and the multi-set $\phi_T(e)$ contains exactly $p$ elements (including repetitions).
    \end{enumerate}
\end{property}
The proof of Property \ref{prop:exchange_p_matchoid} is exactly analagous to
the proof of Proposition 10 in \citet{Feldman2018a}, and we omit it.

\begin{lemma} \label{lemm:ratio_p_matchoid}
    Let $(f, \mathcal{M})$ be an instance of \pmat with optimal solution $O$, and let $A'$, $A$ be produced by Alg.\ref{alg:quick_p_matchoid} on this instance. Then
    $$  f(A \cup O) \leq \frac{(1+\beta)^2 \cdot p }{\beta} f(A')  $$
\end{lemma}

\begin{proof}
    \begin{align*}
        f(A \cup O) - f(A) &= \sum_{e_{j} \in O \setminus A} f(e_{j} \mid A_n \cup \{e_1, \dots, e_{j-1}  \}) \\
                        &\leq \sum_{e_{j} \in O \setminus A} f(e_{j} \mid A_n) \tag{By Submodularity}\\
                        &\leq  \sum_{e_{j} \in O \setminus A} f(e_{j} \mid A_{i(e_j)-1})  \tag{By Submodularity and $A_{i(e_j)-1} \subseteq A_n$} \\
                        &= \sum_{e_j \in O \setminus A} \delta_{e_j} \\
                        &\leq  \sum_{e_j \in O \setminus A} \left( (1+\beta) \cdot  \sum_{a^* \in \phi_O(e_j)} \delta_{a^*}  \right) \tag{By Property \ref{prop:exchange_p_matchoid}.(3)} 
    \end{align*}

    Additionally, 
    \begin{align*}
            & \sum_{e_j \in O \setminus A} \cdot \sum_{a^* \in \phi_O(e_j)} \delta_{a^*}  + p \cdot \sum_{e_j \in O \cap A} \delta_{e_j} \\
         \leq &  \sum_{e_j \in O} \cdot \sum_{a^* \in \phi_O(e_j)} \delta_{a^*} \tag{By Property \ref{prop:exchange_p_matchoid}.(4)}  \\
         \leq & \; p \cdot \sum_{a^* \in A'} \delta_{a^*} +  (p-1) \cdot \sum_{a^* \in A \setminus A'} \delta_{a^*}  \tag{By Property \ref{prop:exchange_p_matchoid}.(1,2)}   \\
         \leq & \; p \cdot \delta(A') +  (p-1) \cdot \delta(A \setminus A') \\
         \leq & \; p \cdot \delta(A') +  \frac{p-1}{\beta} \cdot \delta(A') \tag{By Lemma \ref{lemm:A_minus_Aprime_p_matchoid}} \\
         \leq & \; p \cdot f(A') +  \frac{p-1}{\beta} \cdot f(A') \tag{By Proposition \ref{prop:Aprime_delta_p_matchoid}} \\
         = &  \frac{(1+\beta)\cdot p - 1}{\beta} \cdot f(A')
    \end{align*}
    
    Thus, 
    \begin{align*}
        f(A \cup O) &\leq f(A) +  (1+\beta) \sum_{e \in O \setminus A} \cdot \sum_{a^* \in \phi_O(e)} \delta_{a^*} \\
                    &\leq f(A) +  (1+\beta) \cdot \left[ \frac{(1+\beta)\; p - 1}{\beta} \cdot f(A') -  p \cdot \sum_{e \in O \cap A} \delta_e  \right]  \\
                    &\leq \frac{1 +\beta}{\beta} \cdot f(A') +  (1+\beta) \cdot \left[ \frac{(1+\beta) \; p - 1}{\beta} \cdot f(A') -  p \cdot \sum_{e \in O \cap A} \delta_e  \right]  \tag{By Lemma \ref{lemm:Aprime_A_p_matchoid} } \\
                    &= \frac{(1+\beta)^2 \; p }{\beta} \cdot f(A') -  (1+\beta) \; p \cdot \sum_{e \in O \cap A} \delta_e  \\
                    &\leq \frac{(1+\beta)^2 \; p }{\beta} \cdot f(A')  \tag{By Non-negativity of $f$}
    \end{align*}

\end{proof}

\begin{theorem} \label{theo:quick_p_matchoid}
    The algorithm \ref{alg:quick_p_matchoid} is a $1 / (4p)$-approximation algorithm for \pmat{} with query complexity $n$ exactly.
\end{theorem}

\begin{proof}
    For the query complexity, we only query the function $f$ in line \ref{line:query_p_matchoid} of Alg.\ref{alg:quick_p_matchoid}, the following values of $\delta$ in Alg.\ref{alg:exchange_p_matchoid} do not need to be recomputed. Thus, the query complexity is $n$ exactly.

    For the approximation ratio, observe that:
    $$ 
    f(O) \leq_{(1)} f(O \cup A) \leq_{(2)}  \frac{(1+\beta)^2 \cdot p }{\beta} f(A') 
    $$
    where inequality (1) is for monotonicity and inequality (2) is according to Lemma \ref{lemm:ratio_p_matchoid}. The $1 / (4p) $ ratio follows from optimizing over $\beta \in [0,\infty)$ (the ratio is optimized at $\beta = 1$).
\end{proof}

\section{Experimental results} \label{sec:exp} 

\paragraph{Benchmarks.} We experimentally compare \quickmax{} to other algorithms with low query complexity for \mon{}. The algorithm of \citet{Chakrabarti2015b}, which we refer to as \textbf{CK}, achieves a $1/4$-approximation in $2n$ queries by also making a single pass over the elements. \textbf{Lazy greedy} (\citep{Minoux1978Lazy}) achieves the same $1/2$ approximation and $n \cdot \text{rank}(\mathcal M)$ query complexity as greedy, but achieves a smaller number of queries in practice than greedy by lazily evaluating the marginal contributions. The algorithm of \citet{Badanidiyuru2014b}, which  we refer to as \textbf{threshold greedy}, achieves a $1/2 - \epsi$ approximation in $\mathcal O((n/\epsi) \log(k / \epsi))$ queries by iteratively decreasing a threshold and adding elements with marginal contribution over the current threshold to the current solution. Additional details on these algorithms and their implementation are provided in  Appendix~\ref{sec:appexperiments}. 

Given some fixed ordering with which elements of the ground set are considered, these algorithms are all deterministic. Thus, for each problem instance, we choose 5 random orderings of the ground set, and then run each algorithm on each of these 5 orderings for each matroid, plotting the corresponding average results over these runs.

 \begin{figure*}[t]
	\centering
	\begin{subfigure}[c]{\textwidth}
		\includegraphics[width=0.33\textwidth]{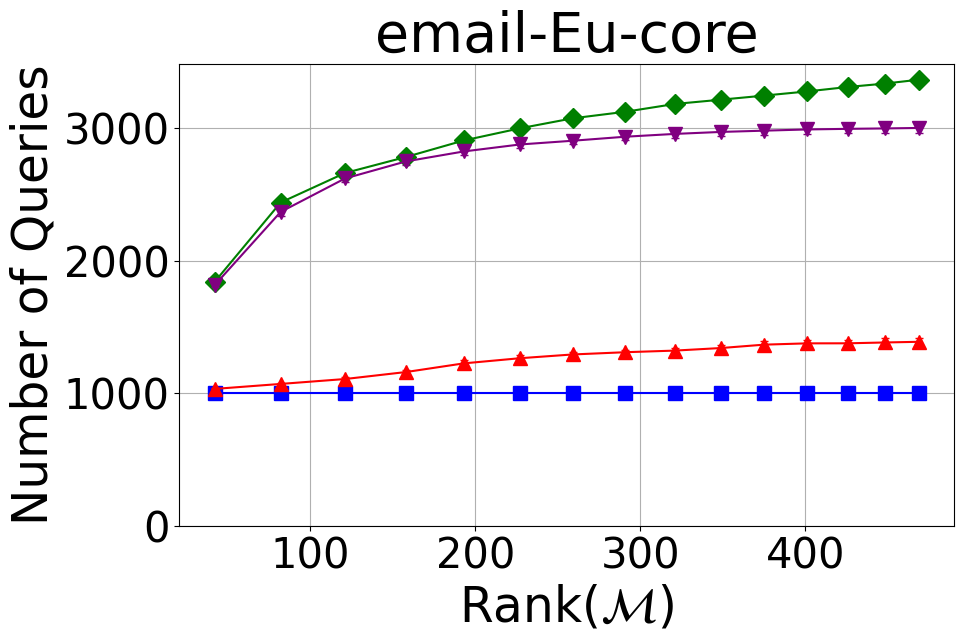}
		\includegraphics[width=0.33\textwidth]{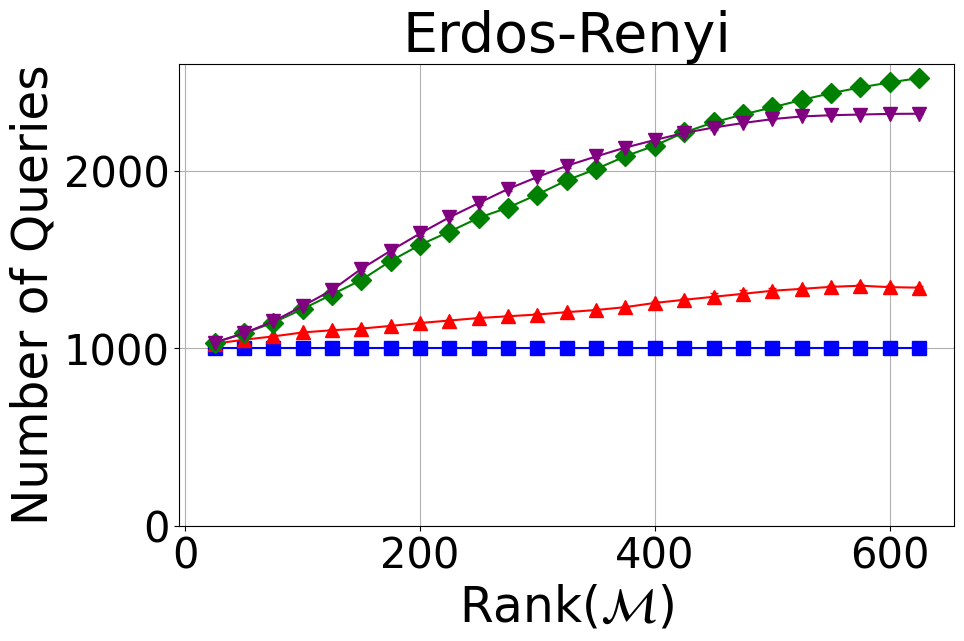}
		\includegraphics[width=0.33\textwidth]{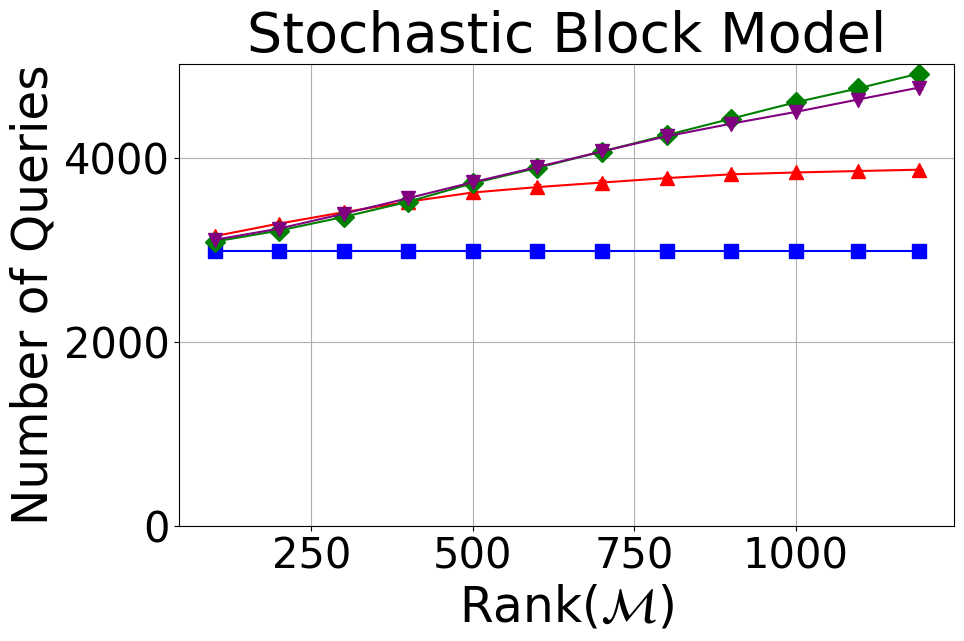}
	\end{subfigure} %
	\begin{subfigure}[c]{\textwidth}
		\includegraphics[width=0.33\textwidth]{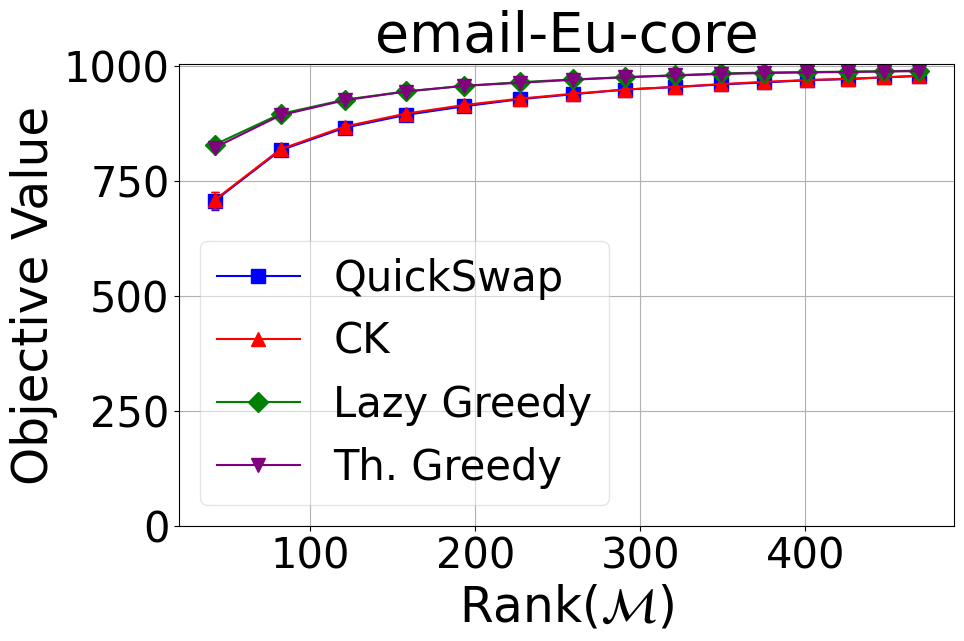}
		\includegraphics[width=0.33\textwidth]{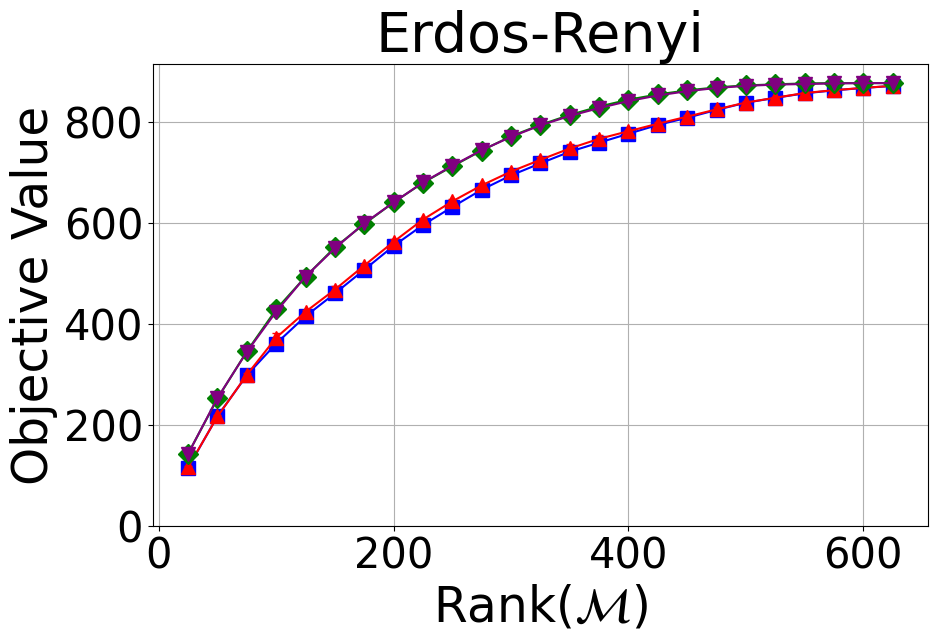}
		\includegraphics[width=0.33\textwidth]{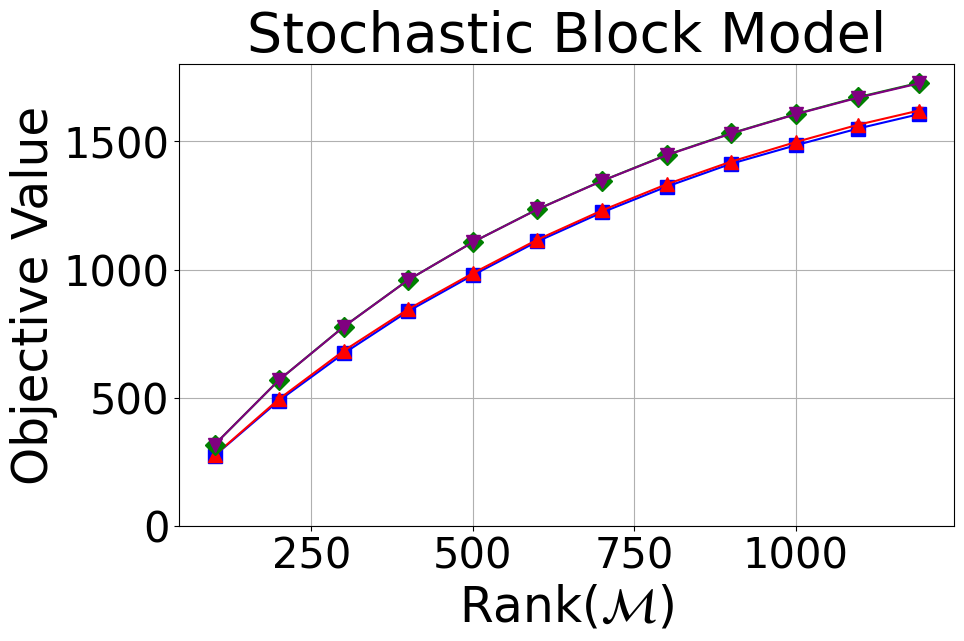}
	\end{subfigure} %
	\caption{The numbers of queries and objective values achieves by the four algorithms for \mon{} as a function of rank of the matroid. For the objective value plots, lazy greedy and threshold greedy achieve nearly identical objective values and their lines overlap. For the same reason, the lines for \quickmax{} and CK also overlap for these plots.}
	\label{fig:InfluenceAlgs}
\end{figure*}

\paragraph{Problem instances.} We consider instances of influence maximization subject to a partition matroid constraint on various graphs, similarly to the experiment setting in \cite{balkanski2018nonmonotone} and \cite{dutting2022deletion}.  In particular, given some graph $G = (V,E)$, we measure influence as defined by the coverage function $f: 2^V \rightarrow Z_{\geq 0}$ where $f(S) = | \{v \in V \, | \, \exists s \in S \text{ such that } (s,v) \in E \} |$. The three graphs over $n$ nodes and $m$ edges, and the partition matroids, are as follows (additional details in Appendix~\ref{sec:appexperiments}):
\begin{itemize}[leftmargin=*]
    \item \textbf{email-Eu-core}: The SNAP network email-Eu-core \citep{email-Eu-core} is a real-world directed social network with $n = 1005$ and $m = 25571$. The dataset also includes which of 41 possible departments each individual belongs to, with these sets forming the parts of the partition matroid.
    \item \textbf{Erdos-Renyi}: An Erdos-Renyi graph with $n = 1000$ and edge-formation probability $p = \frac{1}{500}$. The partition is generated by assigning each node to one of $20$ parts uniformly at random. 
    \item \textbf{Stochastic Block Model}: A randomly generated graph according to the stochastic block model (SBM) with $100$ communities that have uniformly random size between $10$ and $50$.  The parts of the  partition matroid are the 100 sets corresponding to these 100 communities. 
\end{itemize}

\paragraph{Results.} For the objective value plots, lazy greedy and threshold greedy achieve nearly identical objective values and their lines overlap. For the same reason, the lines for \quickmax{} and CK also overlap for these plots.  As shown in Figure \ref{fig:InfluenceAlgs}, \quickmax{} always achieves the least number of queries of all algorithms tested (exactly $n$ in each instance).  \quickmax{} significantly improves on the number of queries compared to  lazy greedy and  threshold greedy, at a small loss in objective value. More precisely, \quickmax{}  always achieves at least 80\%, and often significantly over 90\%, of the objective value achieved by lazy greedy and threshold greedy.  For the performance of \quickmax{} against the CK algorithm, it always requires fewer queries (and typically 20 to 30\% less for larger size constraints), while simultaneously achieving nearly identical objective values.


\section{Conclusion} \label{sec:conclusion}

In this work, we provide a deterministic algorithm for \mon{}
that achieves a constant approximation factor with exactly $n$ queries
to the oracle for the submodular function. However, our ratio of $1/4$
is smaller than the optimal ratio of $1 - 1/e \approx 0.63$.
An interesting question for future work is what is the best
ratio achievable with $\mathcal O(n)$ queries -- in particular, is it possible to
achieve the $1 - 1/e$ ratio, or even the $1/2$ ratio of the greedy algorithm with linear query complexity?
For the general, non-monotone problem \nmon{}, similar questions apply.
We provided the first constant-factor approximation
with linear query complexity, but our approximation ratio of $\approx 0.085$ in this setting is far from the best known ratio of $0.401$ \citep{Buchbinder2023} in polynomial time. 

Furthermore, in this paper we consider the value query model, and measure the efficiency of an algorithm by the number of queries made to the oracle.
We do not consider other metrics, such as the number of arithmetic operations or independence queries to the matroid. While computing the value of $f$ typically
dominates other parts of the computation in most applications of submodular optimization, this may not hold true for all applications. In addition, there are other aspects related to computational efficiency, such as the ability to optimize marginal gain queries. We believe our algorithm would be able to be highly optimized in such settings (as it only queries the marginal gain into a nested sequence of sets), but we did not attempt to evaluate this explicitly. 

\section*{Acknowledgements} Eric Balkanski was supported by NSF grants CCF-2210502 and IIS-2147361. The work of Alan Kuhnle was partially supported by Texas A\&M University. 


\newpage

\bibliographystyle{plainnat} 
\bibliography{alan-refs,References} 

\begin{thebibliography}{32}
\providecommand{\natexlab}[1]{#1}
\providecommand{\url}[1]{\texttt{#1}}
\expandafter\ifx\csname urlstyle\endcsname\relax
  \providecommand{\doi}[1]{doi: #1}\else
  \providecommand{\doi}{doi: \begingroup \urlstyle{rm}\Url}\fi

\bibitem[Ashwinkumar(2011)]{Ashwinkumar2011}
B.~V. Ashwinkumar.
\newblock Buyback {{Problem}} - {{Approximate}} matroid intersection with
  cancellation costs.
\newblock In \emph{Proceedings of the 38th {{International Colloquium
  Conference}} on {{Automata}}, {{Languages}} and {{Programming}}}, volume
  6755, pages 379--390, 2011.
\newblock \doi{10.1007/978-3-642-22006-7_32}.

\bibitem[Badanidiyuru and Vondr{\'a}k(2014)]{Badanidiyuru2014b}
Ashwinkumar Badanidiyuru and Jan Vondr{\'a}k.
\newblock Fast algorithms for maximizing submodular functions.
\newblock In \emph{Proceedings of the {{Twenty-Fifth Annual ACM-SIAM
  Symposium}} on {{Discrete Algorithms}}}, pages 1497--1514. {Society for
  Industrial and Applied Mathematics}, January 2014.
\newblock ISBN 978-1-61197-338-9 978-1-61197-340-2.
\newblock \doi{10.1137/1.9781611973402.110}.

\bibitem[Badanidiyuru and Vondrák(2014)]{accelerated_greedy}
Ashwinkumar Badanidiyuru and Jan Vondrák.
\newblock \emph{Fast algorithms for maximizing submodular functions}, pages
  1497--1514.
\newblock 01 2014.
\newblock ISBN 978-1-61197-338-9.
\newblock \doi{10.1137/1.9781611973402.110}.

\bibitem[Badanidiyuru et~al.(2014)Badanidiyuru, Mirzasoleiman, Karbasi, and
  Krause]{badanidiyuru2014streaming}
Ashwinkumar Badanidiyuru, Baharan Mirzasoleiman, Amin Karbasi, and Andreas
  Krause.
\newblock Streaming submodular maximization: Massive data summarization on the
  fly.
\newblock In \emph{Proceedings of the 20th ACM SIGKDD international conference
  on Knowledge discovery and data mining}, pages 671--680, 2014.

\bibitem[Balkanski et~al.(2018)Balkanski, Breuer, and
  Singer]{balkanski2018nonmonotone}
Eric Balkanski, Adam Breuer, and Yaron Singer.
\newblock Non-monotone submodular maximization in exponentially fewer
  iterations.
\newblock In \emph{Advances in Neural Information Processing Systems 31
  (NeurIPS 2018)}, 2018.

\bibitem[Breuer et~al.(2020)Breuer, Balkanski, and Singer]{Breuer2020}
Adam Breuer, Eric Balkanski, and Yaron Singer.
\newblock The {{FAST Algorithm}} for {{Submodular Maximization}}.
\newblock In \emph{Proceedings of the 37th {{International Conference}} on
  {{Machine Learning}}}, pages 1134--1143. PMLR, November 2020.

\bibitem[Buchbinder and Feldman(2023)]{Buchbinder2023}
Niv Buchbinder and Moran Feldman.
\newblock Constrained {{Submodular Maximization}} via {{New Bounds}} for
  {{DR-Submodular Functions}}, November 2023.

\bibitem[Buchbinder et~al.(2015)Buchbinder, Feldman, and
  Schwartz]{Buchbinder2015}
Niv Buchbinder, Moran Feldman, and Roy Schwartz.
\newblock Comparing {{Apples}} and {{Oranges}}: {{Query Tradeoff}} in
  {{Submodular Maximization}}.
\newblock In \emph{{{ACM-SIAM Symposium}} on {{Discrete Algorithms}}
  ({{SODA}})}, 2015.
\newblock \doi{10.1137/1.9781611973730.77}.

\bibitem[Calinescu et~al.(2011)Calinescu, Chekuri, Pal, and
  Vondr{\'a}k]{calinescu2011maximizing}
Gruia Calinescu, Chandra Chekuri, Martin Pal, and Jan Vondr{\'a}k.
\newblock Maximizing a monotone submodular function subject to a matroid
  constraint.
\newblock \emph{SIAM Journal on Computing}, 40\penalty0 (6):\penalty0
  1740--1766, 2011.

\bibitem[Chakrabarti and Kale(2015)]{Chakrabarti2015b}
Amit Chakrabarti and Sagar Kale.
\newblock Submodular maximization meets streaming: Matchings, matroids, and
  more.
\newblock \emph{Mathematical Programming}, 154\penalty0 (1):\penalty0 225--247,
  December 2015.
\newblock ISSN 1436-4646.
\newblock \doi{10.1007/s10107-015-0900-7}.

\bibitem[Chekuri et~al.(2015)Chekuri, Gupta, and Quanrud]{Chekuri2015}
Chandra Chekuri, Shalmoli Gupta, and Kent Quanrud.
\newblock Streaming {{Algorithms}} for {{Submodular Function Maximization}}.
\newblock In \emph{International {{Colloquium}} on {{Automata}}, {{Languages}},
  and {{Programming}} ({{ICALP}})}, 2015.

\bibitem[Das and Kempe(2011)]{das2011submodular}
Abhimanyu Das and David Kempe.
\newblock Submodular meets spectral: Greedy algorithms for subset selection,
  sparse approximation and dictionary selection.
\newblock In \emph{Proc. 28th Int. Conf. on Machine Learning (ICML'11)}, pages
  1057--1064, 2011.

\bibitem[D{\"u}tting et~al.(2022)D{\"u}tting, Fusco, Lattanzi, Norouzi-Fard,
  and Zadimoghaddam]{dutting2022deletion}
Paul D{\"u}tting, Federico Fusco, Silvio Lattanzi, Ashkan Norouzi-Fard, and
  Morteza Zadimoghaddam.
\newblock Deletion robust submodular maximization over matroids.
\newblock In \emph{Proceedings of the 39th International Conference on Machine
  Learning}, volume 162 of \emph{Proceedings of Machine Learning Research},
  pages 5671--5693. PMLR, 2022.

\bibitem[D{\"u}tting et~al.(2023)D{\"u}tting, Fusco, Lattanzi, Norouzi-Fard,
  and Zadimoghaddam]{dutting2023fully}
Paul D{\"u}tting, Federico Fusco, Silvio Lattanzi, Ashkan Norouzi-Fard, and
  Morteza Zadimoghaddam.
\newblock Fully dynamic submodular maximization over matroids.
\newblock In \emph{International Conference on Machine Learning}, pages
  8821--8835. PMLR, 2023.

\bibitem[Ene and Nguyen(2019)]{Ene2019}
Alina Ene and Huy~L. Nguyen.
\newblock A {{Nearly-linear Time Algorithm}} for {{Submodular Maximization}}
  with a {{Knapsack Constraint}}.
\newblock In \emph{{{ICALP}}}, pages 1--24, 2019.

\bibitem[Feldman et~al.(2018)Feldman, Karbasi, and Kazemi]{Feldman2018a}
Moran Feldman, Amin Karbasi, and Ehsan Kazemi.
\newblock Do {{Less}}, {{Get More}}: {{Streaming Submodular Maximization}} with
  {{Subsampling}}.
\newblock In \emph{Advances in {{Neural Information Processing Systems}}},
  volume~31. Curran Associates, Inc., 2018.

\bibitem[Fisher et~al.(1978)Fisher, Nemhauser, and Wolsey]{fisher1978analysis}
Marshall~L Fisher, George~L Nemhauser, and Laurence~A Wolsey.
\newblock \emph{An analysis of approximations for maximizing submodular set
  functions—II}.
\newblock Springer, 1978.

\bibitem[Goemans()]{unique_circuit}
Michel Goemans.
\newblock Matroid optimization notes.
\newblock \url{https://math.mit.edu/~goemans/18433S13/matroid-notes.pdf}.

\bibitem[Han et~al.(2020)Han, Cao, Cui, and Wu]{Han2020}
Kai Han, Zongmai Cao, Shuang Cui, and Benwei Wu.
\newblock Deterministic {{Approximation}} for {{Submodular Maximization}} over
  a {{Matroid}} in {{Nearly Linear Time}}.
\newblock In \emph{{{NeurIPS}}}, pages 1--12, 2020.

\bibitem[Kempe et~al.(2003)Kempe, Kleinberg, and Tardos]{kempe2003maximizing}
David Kempe, Jon Kleinberg, and {\'E}va Tardos.
\newblock Maximizing the spread of influence through a social network.
\newblock In \emph{Proceedings of the ninth ACM SIGKDD international conference
  on Knowledge discovery and data mining}, pages 137--146, 2003.

\bibitem[Kobayashi and Terao(2024)]{Kobayashi2024}
Yusuke Kobayashi and Tatsuya Terao.
\newblock Subquadratic {{Submodular Maximization}} with a {{General Matroid
  Constraint}}.
\newblock In \emph{{{arXiv}}:2405.00359}. arXiv, May 2024.

\bibitem[Kuhnle(2019)]{Kuhnle2019e}
Alan Kuhnle.
\newblock Interlaced {{Greedy Algorithm}} for {{Maximization}} of {{Submodular
  Functions}} in {{Nearly Linear Time}}.
\newblock In \emph{Advances in {{Neural Information Processing Systems}}},
  volume~32. Curran Associates, Inc., 2019.

\bibitem[Kuhnle(2021)]{Kuhnle2021a}
Alan Kuhnle.
\newblock Quick {{Streaming Algorithms}} for {{Maximization}} of {{Monotone
  Submodular Functions}} in {{Linear Time}}.
\newblock In \emph{Artificial {{Intelligence}} and {{Statistics}}
  ({{AISTATS}})}, 2021.

\bibitem[Kupfer et~al.(2020)Kupfer, Qian, Balkanski, and
  Singer]{kupfer2020adaptive}
Ron Kupfer, Sharon Qian, Eric Balkanski, and Yaron Singer.
\newblock The adaptive complexity of maximizing a gross substitutes valuation.
\newblock \emph{Advances in Neural Information Processing Systems},
  33:\penalty0 19817--19827, 2020.

\bibitem[Lee et~al.(2009)Lee, Mirrokni, Nagarajan, and Sviridenko]{lee2009non}
Jon Lee, Vahab~S Mirrokni, Viswanath Nagarajan, and Maxim Sviridenko.
\newblock Non-monotone submodular maximization under matroid and knapsack
  constraints.
\newblock In \emph{Proceedings of the forty-first annual ACM symposium on
  Theory of computing}, pages 323--332, 2009.

\bibitem[Leskovec and Krevl(2014)]{email-Eu-core}
Jure Leskovec and Andrej Krevl.
\newblock Snap dataset: email-eu-core.
\newblock \url{https://snap.stanford.edu/data/email-Eu-core.html}, 2014.

\bibitem[Li et~al.(2022)Li, Feldman, Kazemi, and Karbasi]{Li2022}
Wenxin Li, Moran Feldman, Ehsan Kazemi, and Amin Karbasi.
\newblock Submodular {{Maximization}} in {{Clean Linear Time}}.
\newblock \emph{Advances in Neural Information Processing Systems},
  35:\penalty0 17473--17487, December 2022.

\bibitem[Lin and Bilmes(2011)]{lin2011class}
Hui Lin and Jeff Bilmes.
\newblock A class of submodular functions for document summarization.
\newblock In \emph{Proceedings of the 49th annual meeting of the association
  for computational linguistics: human language technologies}, pages 510--520,
  2011.

\bibitem[Minoux(1978)]{Minoux1978Lazy}
Michel Minoux.
\newblock Accelerated greedy algorithms for maximizing submodular set
  functions.
\newblock In J.~Stoer, editor, \emph{Optimization Techniques}, pages 234--243,
  Berlin, Heidelberg, 1978. Springer Berlin Heidelberg.

\bibitem[Mirzasoleiman et~al.(2015)Mirzasoleiman, Badanidiyuru, Karbasi,
  Vondrak, and Krause]{Mirzasoleiman2015}
Baharan Mirzasoleiman, Ashwinkumar Badanidiyuru, Amin Karbasi, Jan Vondrak, and
  Andreas Krause.
\newblock Lazier {{Than Lazy Greedy}}.
\newblock In \emph{{{AAAI Conference}} on {{Artificial Intelligence}}
  ({{AAAI}})}, 2015.
\newblock ISBN 978-1-57735-701-8.

\bibitem[Nemhauser and Wolsey(1978)]{nemhauser1978best}
George~L Nemhauser and Laurence~A Wolsey.
\newblock Best algorithms for approximating the maximum of a submodular set
  function.
\newblock \emph{Mathematics of operations research}, 3\penalty0 (3):\penalty0
  177--188, 1978.

\bibitem[Norouzi-Fard et~al.(2018)Norouzi-Fard, Tarnawski, Mitrovic, Zandieh,
  Mousavifar, and Svensson]{norouzi2018beyond}
Ashkan Norouzi-Fard, Jakub Tarnawski, Slobodan Mitrovic, Amir Zandieh,
  Aidasadat Mousavifar, and Ola Svensson.
\newblock Beyond 1/2-approximation for submodular maximization on massive data
  streams.
\newblock In \emph{International Conference on Machine Learning}, pages
  3829--3838. PMLR, 2018.

\end{thebibliography}

\appendix



\section{Additional details about experimental setup}
\label{sec:appexperiments}

 \subsection{Additional details about the algorithms}
 \label{app_algos}
 
Below we provide specific details regarding the implementations of our algorithms from the experimental setup, including how queries were counted as well as specific data structures and parameter settings used. Note that in all implementations, once a given set has been queried for its value and causes the query count to increase by 1, this set will never again contribute to an increase in query count for any repeated computation which requires its evaluation.

\begin{itemize}
    \item \textbf{\quickmax{}:} Implemented exactly as in Algorithm \ref{algo:original-matroid}, with $\beta = 1$. 

    \item \textbf{CK Algorithm:} This algorithm by \citet{Chakrabarti2015b} takes one pass through the ground set and maintains a feasible
set through the following swapping logic. Each element is assigned a weight (which requires at
most two queries to the oracle), and the feasible solution is updated via appealing to an algorithm
of \citet{Ashwinkumar2011} for maximum (modular) weight independent set. In particular, the weight of an element $e$ at the time of its consideration is set as its marginal contribution to the current maintained feasible, and it maintains this weight assignment throughout the entire duration of the algorithm. If $e$ can be added to this feasible set while maintaining its feasibility, it is added. Otherwise, it checks to see if the weight of $e$ is at least twice the weight of the minimum-weight element in the current feasible set which can be swapped with $e$ while maintaining feasibility. If this is the case, the swap occurs, and otherwise $e$ is not added to the feasible set and it remains unchanged. 

    \item \textbf{Lazy Greedy:} Initially, the value of each singelton set is queried, with these corresponding values being assigned as priorties to each corresponding element and pushed onto a max heap. The initial query count is thus set to $n$, and the top element of the heap is popped and added to the maintained independent set $S$. From this point on, in each iteration the top element $e$ is popped from the max heap and we check if $e$ is the same element that was most recently popped. If so, we add $e$ to $S$. Otherwise, we check if $S+e$ is independent.  If so, $e$ is pushed back onto the heap with priority $f(e | S)$ (and the query count is incremented by 1, not 2, as $f(S)$ must have already been queried as this computation only requires we newly query $f(S+e)$), and otherwise $e$ is not pushed onto the heap nor queried for its marginal contribution and we continue. This process continues until the heap is empty, i.e. no new element can be added to $S$ while maintaining its independence.  

    \item \textbf{Threshold Greedy:} Two instances of threshold greedy were tested, for corresponding values of $\epsi = 0.1$ and $\epsi = \frac{1}{6}$. The value $\epsi = \frac{1}{6}$ was chosen as it corresponds to the threshold value guaranteeing a $\frac{1}{4}$-approximation ratio. The value $\epsi = 0.1$ was chosen to improve solution quality, at the potential risk of increasing the number of required queries. As these two parameters led to nearly identical performance (visually indistinguishable behavior in the plots), the version of threshold greedy implemented in this paper is for $\epsi = \frac{1}{6}$. Our implementation is based on that of Algorithm 1 in \citet{accelerated_greedy}, which is a threshold greedy algorithm for monotone submodular maximization subject to a cardinality constraint, except that in each instance the algorithm checks $S+e$ satisfies the cardinality constraint we instead check that the appropriate matroid constraint is satisfied, and the lower-limit on the threshold at which point we stop checking is set to $\frac{\epsi}{r}$ where $r$ is the rank of the matroid. Furthermore, when we count queries, for each element we lazily store its last queried marginal contribution, so that it is only reevaluated should it be at least as large as the current threshold (as otherwise submodularity guarantees it must still be less than the current threshold). 
\end{itemize}

 \subsection{Additional details about the problem instances} \label{app_problems}

  For each graph $G$, the partition matroid constraints considered are defined by a fixed partition of $V$ into sets $P_1, \ldots, P_m$ (whose union is $V$ and pairwise intersection is always empty) and corresponding nonnegative integers $k_1, \ldots, k_m$ such that any set $I \subset V$ is independent if and only if $I \cap P_i \leq k_i$ for each $i$. For our experiments, we consider partition matroids with fixed size constraints, i.e. such that $k_1 = \ldots = k_m$. For the instances, we consider partition matroids with positive integer size constraints from $1$ up to $15$, $25$ and $12$ respectively.  The three graphs  and the partition matroids are as follows:
\begin{itemize}[leftmargin=*]
    \item \textbf{email-Eu-core}: The SNAP network email-Eu-core \citep{email-Eu-core} representing email exchanges between members of a large European research institution, where each node corresponds to a member of the institution and each edge points from email senders to their recipients. The dataset also includes which of 41 possible departments each individual belongs to, with these sets forming the parts of the partition underlying each partition matroid considered. There are 1005 nodes and 25571 edges in this graph. Note that this dataset is released under the BSD license on SNAP, meaning it is free for both academic and commercial use.
    \item \textbf{Erdos-Renyi}: An Erdos-Renyi graph of 1000 nodes with edge-formation probability $p = \frac{1}{500}$, and a partition generated by assigning each node to one of 20 parts uniformly at random. 
    \item \textbf{Stochastic Block Model}: A randomly generated graph according to the stochastic block model (SBM), with 100 communities set to have sizes chosen uniformly random from 10 to 50 inclusive, with edge formation probability of 0 between nodes in different communities and $\frac{1}{30}$ between nodes in the same community. The parts of the underlying partition are the 100 sets corresponding to these 100 communities. 
\end{itemize}

\subsection{Mean ± standard deviation tables} \label{stdev_appendix}

As each data point plotted corresponds to a mean of 5 trials, the corresponding standard deviations were computed. However, each was too small to be visible on the plots, with most standard deviations being less than 1\% of their corresponding mean. Thus, these values, along with the exact value of the mean, are provided in Tables 1, 2, 3, 4, 5 and 6.

\subsection{Compute resources} \label{compute_appendix}

All code was run in a Jupyter Notebook, and took no more than 3 hours total to run on a standard Mac with an M1 chip.

\begin{table}[ht]
\centering
\caption{Mean ± Standard Deviation of Queries on email-Eu-core}
\begin{adjustbox}{width=0.7\textwidth}
\begin{tabular}{|c|c|c|c|c|}
\hline
Matroid Rank & QuickSwap & CK & Lazy Greedy & Threshold Greedy \\
\hline
42 & 1005.0 ± 0.0 & 1032.4 ± 3.1 & 1840.0 ± 0.0 & 1815.4 ± 15.5 \\
\hline
82 & 1005.0 ± 0.0 & 1070.0 ± 7.8 & 2439.0 ± 0.0 & 2368.0 ± 33.4 \\
\hline
121 & 1005.0 ± 0.0 & 1107.0 ± 12.0 & 2661.6 ± 0.5 & 2620.4 ± 25.6 \\
\hline
158 & 1005.0 ± 0.0 & 1159.8 ± 15.7 & 2782.8 ± 1.2 & 2749.4 ± 25.0 \\
\hline
193 & 1005.0 ± 0.0 & 1225.2 ± 29.0 & 2907.2 ± 0.7 & 2824.4 ± 29.5 \\
\hline
227 & 1005.0 ± 0.0 & 1264.0 ± 21.7 & 2998.8 ± 1.3 & 2877.0 ± 26.6 \\
\hline
259 & 1005.0 ± 0.0 & 1292.4 ± 22.1 & 3074.4 ± 1.6 & 2905.2 ± 26.5 \\
\hline
291 & 1005.0 ± 0.0 & 1309.4 ± 16.5 & 3124.0 ± 1.1 & 2935.6 ± 26.5 \\
\hline
321 & 1005.0 ± 0.0 & 1321.6 ± 18.0 & 3182.8 ± 0.7 & 2956.8 ± 30.6 \\
\hline
349 & 1005.0 ± 0.0 & 1341.0 ± 21.3 & 3215.2 ± 2.0 & 2971.8 ± 32.6 \\
\hline
375 & 1005.0 ± 0.0 & 1366.4 ± 24.7 & 3245.8 ± 0.7 & 2981.2 ± 30.9 \\
\hline
401 & 1005.0 ± 0.0 & 1376.2 ± 24.4 & 3277.2 ± 1.3 & 2990.4 ± 32.5 \\
\hline
426 & 1005.0 ± 0.0 & 1376.6 ± 28.2 & 3311.2 ± 0.4 & 2995.0 ± 34.1 \\
\hline
448 & 1005.0 ± 0.0 & 1382.8 ± 30.7 & 3336.2 ± 0.4 & 2998.0 ± 34.2 \\
\hline
469 & 1005.0 ± 0.0 & 1388.0 ± 28.3 & 3366.2 ± 2.0 & 3001.8 ± 35.8 \\
\hline

\hline
\end{tabular}
\end{adjustbox}
\label{tab:std_dev_queries}
\end{table}

\begin{table}[ht]
\centering
\caption{Mean ± Standard Deviation of Queries on ER}
\begin{adjustbox}{width=0.7\textwidth}
\begin{tabular}{|c|c|c|c|c|}
\hline
Matroid Rank & QuickSwap & CK & Lazy Greedy & Threshold Greedy \\
\hline
25 & 1000.0 ± 0.0 & 1026.6 ± 2.7 & 1032.4 ± 1.9 & 1033.0 ± 2.6 \\
\hline
50 & 1000.0 ± 0.0 & 1049.2 ± 3.6 & 1084.8 ± 3.4 & 1085.8 ± 4.1 \\
\hline
75 & 1000.0 ± 0.0 & 1067.8 ± 7.0 & 1145.8 ± 5.2 & 1152.6 ± 6.4 \\
\hline
100 & 1000.0 ± 0.0 & 1090.2 ± 8.0 & 1222.2 ± 7.9 & 1238.8 ± 14.6 \\
\hline
125 & 1000.0 ± 0.0 & 1102.6 ± 6.9 & 1304.2 ± 5.8 & 1329.0 ± 10.4 \\
\hline
150 & 1000.0 ± 0.0 & 1112.0 ± 8.2 & 1386.0 ± 9.6 & 1449.0 ± 14.1 \\
\hline
175 & 1000.0 ± 0.0 & 1127.0 ± 5.8 & 1493.2 ± 7.0 & 1551.0 ± 15.7 \\
\hline
200 & 1000.0 ± 0.0 & 1143.4 ± 4.5 & 1584.0 ± 12.1 & 1648.4 ± 14.4 \\
\hline
225 & 1000.0 ± 0.0 & 1157.4 ± 1.9 & 1658.6 ± 16.7 & 1739.6 ± 11.0 \\
\hline
250 & 1000.0 ± 0.0 & 1171.6 ± 6.2 & 1736.8 ± 9.9 & 1820.0 ± 9.8 \\
\hline
275 & 1000.0 ± 0.0 & 1181.4 ± 8.1 & 1794.6 ± 12.4 & 1899.6 ± 10.5 \\
\hline
300 & 1000.0 ± 0.0 & 1191.4 ± 9.0 & 1868.2 ± 15.7 & 1967.0 ± 13.3 \\
\hline
325 & 1000.0 ± 0.0 & 1204.6 ± 7.5 & 1948.4 ± 10.0 & 2030.2 ± 14.0 \\
\hline
350 & 1000.0 ± 0.0 & 1216.8 ± 6.6 & 2011.8 ± 14.2 & 2083.0 ± 9.1 \\
\hline
375 & 1000.0 ± 0.0 & 1233.0 ± 7.8 & 2084.4 ± 7.1 & 2133.2 ± 10.9 \\
\hline
400 & 1000.0 ± 0.0 & 1256.8 ± 8.7 & 2144.2 ± 11.9 & 2175.6 ± 15.6 \\
\hline
425 & 1000.0 ± 0.0 & 1274.4 ± 10.3 & 2218.6 ± 14.4 & 2215.6 ± 15.3 \\
\hline
450 & 1000.0 ± 0.0 & 1291.4 ± 19.4 & 2275.0 ± 11.4 & 2245.8 ± 13.4 \\
\hline
475 & 1000.0 ± 0.0 & 1308.0 ± 19.7 & 2319.6 ± 10.7 & 2271.4 ± 6.7 \\
\hline
500 & 1000.0 ± 0.0 & 1325.4 ± 17.0 & 2359.4 ± 12.7 & 2293.4 ± 7.6 \\
\hline
525 & 1000.0 ± 0.0 & 1336.0 ± 13.0 & 2401.8 ± 9.3 & 2308.6 ± 5.5 \\
\hline
550 & 1000.0 ± 0.0 & 1347.4 ± 13.2 & 2439.8 ± 6.8 & 2315.4 ± 5.3 \\
\hline
575 & 1000.0 ± 0.0 & 1353.8 ± 11.6 & 2472.2 ± 3.3 & 2319.4 ± 5.0 \\
\hline
600 & 1000.0 ± 0.0 & 1345.8 ± 12.0 & 2500.2 ± 4.3 & 2323.0 ± 3.4 \\
\hline
625 & 1000.0 ± 0.0 & 1342.8 ± 10.2 & 2525.2 ± 4.3 & 2324.0 ± 3.6 \\
\hline

\hline
\end{tabular}
\end{adjustbox}
\label{tab:std_dev_queriestwo}
\end{table}

\begin{table}[ht]
\centering
\caption{Mean ± Standard Deviation of Queries on SBM}
\begin{adjustbox}{width=0.7\textwidth}
\begin{tabular}{|c|c|c|c|c|}
\hline
Matroid Rank & QuickSwap & CK & Lazy Greedy & Threshold Greedy \\
\hline
100 & 2992.0 ± 0.0 & 3150.0 ± 8.9 & 3091.0 ± 0.0 & 3110.0 ± 0.0 \\
\hline
200 & 2992.0 ± 0.0 & 3287.0 ± 16.2 & 3210.8 ± 2.5 & 3232.4 ± 1.4 \\
\hline
300 & 2992.0 ± 0.0 & 3410.8 ± 7.6 & 3359.4 ± 2.0 & 3394.2 ± 3.5 \\
\hline
400 & 2992.0 ± 0.0 & 3525.2 ± 13.1 & 3524.8 ± 8.1 & 3563.4 ± 6.6 \\
\hline
500 & 2992.0 ± 0.0 & 3625.2 ± 25.0 & 3726.0 ± 6.3 & 3736.8 ± 13.1 \\
\hline
600 & 2992.0 ± 0.0 & 3684.0 ± 33.0 & 3892.6 ± 6.0 & 3902.6 ± 7.5 \\
\hline
700 & 2992.0 ± 0.0 & 3733.6 ± 29.5 & 4071.2 ± 4.1 & 4071.8 ± 5.3 \\
\hline
800 & 2992.0 ± 0.0 & 3781.6 ± 25.0 & 4249.2 ± 13.5 & 4236.6 ± 8.5 \\
\hline
900 & 2992.0 ± 0.0 & 3823.0 ± 23.7 & 4428.0 ± 19.0 & 4373.6 ± 15.0 \\
\hline
1000 & 2992.0 ± 0.0 & 3842.8 ± 14.7 & 4606.6 ± 15.7 & 4501.8 ± 11.8 \\
\hline
1096 & 2992.0 ± 0.0 & 3858.8 ± 11.1 & 4757.0 ± 12.4 & 4636.0 ± 11.4 \\
\hline
1191 & 2992.0 ± 0.0 & 3873.2 ± 8.6 & 4919.2 ± 10.1 & 4766.8 ± 13.6 \\
\hline

\hline
\end{tabular}
\end{adjustbox}
\label{tab:std_dev_queriesthree}
\end{table}

\begin{table}[ht]
\centering
\caption{Mean ± Standard Deviation of Objective Values on email-Eu-core}
\begin{adjustbox}{width=0.7\textwidth}
\begin{tabular}{|c|c|c|c|c|}
\hline
Matroid Rank & QuickSwap & CK & Lazy Greedy & Threshold Greedy \\
\hline
42 & 706.6 ± 20.5 & 708.4 ± 18.7 & 829.0 ± 0.0 & 822.4 ± 4.5 \\
\hline
82 & 817.0 ± 6.1 & 819.2 ± 4.9 & 896.0 ± 0.0 & 893.0 ± 1.3 \\
\hline
121 & 866.0 ± 4.2 & 868.6 ± 4.4 & 927.0 ± 0.0 & 925.8 ± 1.7 \\
\hline
158 & 893.8 ± 2.5 & 896.6 ± 3.1 & 945.0 ± 0.0 & 944.6 ± 1.4 \\
\hline
193 & 912.4 ± 3.8 & 915.2 ± 4.8 & 957.0 ± 0.0 & 957.0 ± 0.6 \\
\hline
227 & 928.0 ± 5.4 & 929.2 ± 6.4 & 965.0 ± 0.0 & 963.8 ± 1.5 \\
\hline
259 & 938.6 ± 2.9 & 939.8 ± 4.0 & 971.0 ± 0.0 & 970.2 ± 1.0 \\
\hline
291 & 949.0 ± 2.9 & 949.0 ± 3.0 & 976.0 ± 0.0 & 975.8 ± 1.7 \\
\hline
321 & 954.4 ± 3.1 & 954.8 ± 2.9 & 980.0 ± 0.0 & 979.6 ± 2.1 \\
\hline
349 & 960.0 ± 3.1 & 960.8 ± 2.9 & 984.0 ± 0.0 & 982.8 ± 1.7 \\
\hline
375 & 964.4 ± 4.1 & 965.8 ± 4.3 & 986.0 ± 0.0 & 985.0 ± 1.1 \\
\hline
401 & 969.2 ± 3.9 & 969.6 ± 4.2 & 987.0 ± 0.0 & 986.2 ± 0.4 \\
\hline
426 & 971.8 ± 2.6 & 972.4 ± 3.3 & 988.0 ± 0.0 & 987.2 ± 0.7 \\
\hline
448 & 975.4 ± 2.8 & 975.6 ± 2.9 & 989.0 ± 0.0 & 988.2 ± 0.7 \\
\hline
469 & 978.8 ± 2.2 & 979.0 ± 2.3 & 990.0 ± 0.0 & 989.0 ± 0.9 \\
\hline

\hline
\end{tabular}
\end{adjustbox}
\label{tab:std_dev_values}
\end{table}

\begin{table}[ht]
\centering
\caption{Mean ± Standard Deviation of Objective Values on ER}
\begin{adjustbox}{width=0.7\textwidth}
\begin{tabular}{|c|c|c|c|c|}
\hline
Matroid Rank & QuickSwap & CK & Lazy Greedy & Threshold Greedy \\
\hline
25 & 114.8 ± 1.9 & 116.0 ± 2.8 & 142.6 ± 1.0 & 142.4 ± 1.0 \\
\hline
50 & 216.8 ± 3.8 & 216.2 ± 3.9 & 251.8 ± 1.2 & 253.0 ± 0.6 \\
\hline
75 & 297.8 ± 3.1 & 298.6 ± 3.8 & 345.0 ± 1.1 & 344.0 ± 1.7 \\
\hline
100 & 360.6 ± 5.7 & 372.2 ± 9.0 & 428.2 ± 0.4 & 423.8 ± 1.0 \\
\hline
125 & 414.6 ± 2.8 & 423.4 ± 2.7 & 492.8 ± 1.7 & 493.2 ± 2.0 \\
\hline
150 & 459.8 ± 7.5 & 467.4 ± 6.9 & 551.0 ± 4.4 & 549.0 ± 3.1 \\
\hline
175 & 506.6 ± 5.4 & 515.0 ± 2.6 & 597.2 ± 2.5 & 598.6 ± 2.8 \\
\hline
200 & 553.8 ± 6.5 & 562.0 ± 5.9 & 640.0 ± 2.3 & 639.8 ± 3.2 \\
\hline
225 & 595.2 ± 5.8 & 606.0 ± 3.6 & 678.8 ± 2.1 & 679.4 ± 2.4 \\
\hline
250 & 630.8 ± 6.0 & 642.0 ± 4.0 & 711.2 ± 2.0 & 712.4 ± 1.9 \\
\hline
275 & 665.0 ± 5.9 & 673.6 ± 3.0 & 742.6 ± 1.9 & 743.2 ± 1.0 \\
\hline
300 & 693.4 ± 4.4 & 701.0 ± 3.7 & 770.8 ± 1.9 & 769.6 ± 1.4 \\
\hline
325 & 717.2 ± 4.2 & 724.2 ± 6.5 & 793.4 ± 2.1 & 792.4 ± 1.4 \\
\hline
350 & 740.2 ± 7.7 & 746.6 ± 8.8 & 813.0 ± 3.8 & 811.0 ± 2.6 \\
\hline
375 & 758.2 ± 8.2 & 765.8 ± 7.4 & 829.6 ± 3.3 & 826.6 ± 4.3 \\
\hline
400 & 775.8 ± 6.0 & 780.8 ± 3.2 & 843.2 ± 3.2 & 840.4 ± 4.0 \\
\hline
425 & 793.2 ± 5.2 & 795.6 ± 4.4 & 853.6 ± 1.6 & 851.6 ± 3.7 \\
\hline
450 & 807.4 ± 4.8 & 809.6 ± 3.9 & 861.6 ± 1.5 & 860.0 ± 3.0 \\
\hline
475 & 823.0 ± 3.9 & 824.0 ± 3.9 & 867.6 ± 2.1 & 866.4 ± 3.1 \\
\hline
500 & 837.2 ± 3.7 & 837.6 ± 3.1 & 871.4 ± 1.0 & 870.6 ± 2.3 \\
\hline
525 & 846.8 ± 2.0 & 847.0 ± 1.9 & 873.6 ± 0.5 & 873.0 ± 2.1 \\
\hline
550 & 855.8 ± 2.6 & 856.4 ± 2.7 & 875.4 ± 0.5 & 874.2 ± 1.5 \\
\hline
575 & 861.6 ± 2.8 & 862.0 ± 2.6 & 876.0 ± 0.0 & 875.4 ± 0.8 \\
\hline
600 & 866.4 ± 1.5 & 866.6 ± 1.4 & 876.0 ± 0.0 & 876.0 ± 0.0 \\
\hline
625 & 870.4 ± 1.4 & 870.6 ± 1.5 & 876.0 ± 0.0 & 876.0 ± 0.0 \\
\hline

\hline
\end{tabular}
\end{adjustbox}
\label{tab:std_dev_valuestwo}
\end{table}

\begin{table}[ht]
\centering
\caption{Mean ± Standard Deviation of Objective Values on SBM}
\begin{adjustbox}{width=0.7\textwidth}
\begin{tabular}{|c|c|c|c|c|}
\hline
Matroid Rank & QuickSwap & CK & Lazy Greedy & Threshold Greedy \\
\hline
100 & 273.2 ± 3.5 & 274.4 ± 3.0 & 316.0 ± 0.0 & 316.0 ± 0.0 \\
\hline
200 & 488.6 ± 8.5 & 494.8 ± 7.2 & 569.0 ± 0.0 & 569.0 ± 0.0 \\
\hline
300 & 672.4 ± 9.8 & 682.6 ± 11.9 & 777.0 ± 0.0 & 776.2 ± 0.4 \\
\hline
400 & 837.4 ± 4.6 & 845.4 ± 9.0 & 959.6 ± 1.0 & 959.4 ± 1.0 \\
\hline
500 & 978.0 ± 8.2 & 985.4 ± 9.1 & 1107.0 ± 0.9 & 1106.8 ± 1.3 \\
\hline
600 & 1109.6 ± 9.5 & 1116.6 ± 8.3 & 1235.4 ± 1.0 & 1234.6 ± 1.9 \\
\hline
700 & 1222.8 ± 6.5 & 1230.6 ± 6.8 & 1346.2 ± 1.2 & 1345.6 ± 2.2 \\
\hline
800 & 1323.4 ± 4.5 & 1332.8 ± 2.8 & 1447.4 ± 1.2 & 1446.0 ± 2.4 \\
\hline
900 & 1413.0 ± 3.0 & 1421.2 ± 3.2 & 1532.0 ± 0.6 & 1530.6 ± 2.3 \\
\hline
1000 & 1485.0 ± 8.3 & 1497.4 ± 5.3 & 1607.0 ± 1.1 & 1605.8 ± 2.3 \\
\hline
1096 & 1550.0 ± 10.3 & 1564.8 ± 10.4 & 1672.4 ± 1.6 & 1670.4 ± 1.7 \\
\hline
1191 & 1606.4 ± 10.4 & 1620.6 ± 13.3 & 1728.8 ± 1.5 & 1726.6 ± 1.6 \\
\hline

\hline
\end{tabular}
\end{adjustbox}
\label{tab:std_dev_valuesthree}
\end{table}

\end{document}